\def\llncs{0}

\ifnum\llncs=1
\documentclass{llncs}
\else
\documentclass{article}
\usepackage{fullpage}
\fi
\usepackage{amsmath}

\usepackage{amsthm}
\usepackage{amssymb}
\usepackage{physics}
\usepackage{hyperref}
\usepackage{cleveref}
\usepackage{dsfont}
\usepackage{xcolor}
\usepackage{mathtools}
\usepackage{breakcites}

\ifnum\llncs = 0
\newtheorem{theorem}{Theorem}[section]

\newtheorem{corollary}[theorem]{Corollary}
\newtheorem{definition}[theorem]{Definition}
\newtheorem{lemma}[theorem]{Lemma} % for llncs
\newtheorem{proposition}[theorem]{Proposition}

\theoremstyle{definition}

\theoremstyle{remark}
\newtheorem{remark}[theorem]{Remark}

\fi
\newtheorem{notation}{Notation}

\newtheorem{construction}{Construction}

\newcommand{\A}{\mathcal{A}}
\newcommand{\B}{\mathcal{B}}
\newcommand{\D}{\mathcal{D}}
\renewcommand{\O}{\mathcal{O}}
\newcommand{\N}{\mathbb{N}}

\newcommand{\C}{\mathbb{C}}
\newcommand{\E}{\mathop{\mathbb{E}}}
\newcommand{\Ind}{\mathds{1}}

\newcommand{\poly}{poly}

\newcommand{\wt}[1]{\widetilde{#1}}
\newcommand{\ol}[1]{\overline{#1}}

\DeclareMathOperator{\Sim}{Sim}
\DeclareMathOperator{\row}{RowSpace}
\newcommand{\SSM}{ \textsc{\textbf{CHRS}}}
\newcommand{\SSMM}{\textsc{\textbf{CHRS-}}}
\newcommand{\Swap}{\textsc{\textbf{Swap}}}

\newcommand{\Zero}{\mathfrak{Z}}

\newcommand{\Rep}[1]{\smash[b]{\underset{#1}{Rep}}}
\newcommand{\Set}[1]{\smash[b]{\underset{#1}{Set}}}
\usepackage{graphicx} % Required for inserting images
\title{Translating Between the Common Haar Random \\State Model and the Unitary Model}
\ifnum\llncs=1
\author{}
\date{}
\institute{}
\else
\author{
    Eli Goldin\thanks{Email: \texttt{eli.goldin@nyu.edu}. Partially supported by a National Science Foundation Graduate Research Fellowship.}
    \vspace{0.2cm}
    \\{\small New York University\vspace{0.3cm}}
    \and
    Mark Zhandry\thanks{Email: \texttt{mzhandry@gmail.com}}
    \vspace{0.2cm}
    \\{\small NTT Research\vspace{0.3cm}}
}
\fi

\begin{document}

\maketitle

\begin{abstract}Black-box separations are a cornerstone of cryptography, indicating barriers to various goals. A recent line of work has explored black-box separations for \emph{quantum} cryptographic primitives. Namely, a number of separations are known in the Common Haar Random State (CHRS) model, though this model is not considered a complete separation, but rather a starting point. A few very recent works have attempted to lift these separations to a unitary separation, which are considered complete separations. Unfortunately, we find significant errors in some of these lifting results.

We prove general conditions under which CHRS separations can be generically lifted, thereby giving simple, modular, and bug-free proofs of complete unitary separations between various quantum primitives. Our techniques allow for simpler proofs of existing separations as well as new separations that were previously only known in the CHRS model.
\end{abstract}

\newpage
\tableofcontents
\newpage

\section{Introduction}\label{sec:intro}

As cryptography transitions into a quantum world, it is important to understand the relationship between various quantum cryptographic concepts. In particular, while one-way functions are widely considered to be the most basic \emph{classical} object, it has recently become apparent that there is a vast world of primitives that live even below one-way functions. To help understand this new world (often dubbed ``MicroCrypt''), there has recently been much work on showing black-box separations between various tasks.

In this work, we will focus on a recently popular approach, which is the common Haar-random state (CHRS) model~\cite{ARXIV:CheColSat24,TCC:AnaGulLin24}. This model gives all users -- protocols and adversaries alike -- many copies of a single Haar-random quantum state. Some primitives such as quantum commitments, one-way state puzzles, and 1-copy pseudorandom states (1-PRS) exist in this model, but several recent works~\cite{ARXIV:CheColSat24,TCC:AnaGulLin24,EPRINT:AnaGulLin24a} have shown that a number of other primitives do \emph{not} exist, such as many-time PRSs, one-way state generators, and commitments and key agreement with \emph{classical} communication. Thus, these primitives are all separated from quantum commitments/one-way puzzles/1-PRSs in this model.

The power of the CHRS model comes from its simplicity, as it is essentially the most basic idealized model one can imagine. However, the CHRS model is an \emph{isometry} mapping a small input state (in this case, the empty state) into a larger output state. Importantly, the isometry is \emph{irreversible}, and there is no way to coherently eliminate a copy of the state. Such an oracle does not adequately reflect ``real-world'' techniques that are available, as quantum (unitary) circuits are reversible. For this reason, a separation in the CHRS model is not considered a full separation, though it is still useful as a starting point. Instead, it is preferable to use a \emph{unitary} (reversible) oracle.

\begin{remark}One may ask if it is even better to use a \emph{classical} oracle for separations. Such a classical oracle separation would capture systems based off of classical functions, and would indeed be relevant for separating primitives that live ``above'' one-way functions. However, for the inherently quantum MicroCrypt primitives we focus on here, the primitives in general cannot be represented using classical functions, but instead will be represented using unitary operations. Therefore, a unitary separation seems essentially as good as a classical one.
\end{remark}

Unfortunately, unitary oracles are much harder to work with, as now one has to reason about query complexity and deal with issues like reversing computation or adaptivity. Very recently, several works have started to ``lift'' the CHRS model separations into a unitary separation~\cite{ARXIV:CheColSat24,EPRINT:BosCheNeh24,EPRINT:BMMMY24}\footnote{The first version of~\cite{ARXIV:CheColSat24} only included the CHRS model, but was more recently updated to include extensions to the unitary model.} to give a unitary black-box separation. The main result common to these works is to show that quantum commitments do not imply many-time PRSs, though the different works have different variants of this statement. These works all employ somewhat similar ideas, but differ significantly on the underlying details.

\paragraph{Our Work.} Our main result, which is inspired by the particular techniques of~\cite{EPRINT:BMMMY24}, is the following:
\begin{theorem}[informal]\label{thm:main_inf} For any 
``typical''\footnote{Here, a ``typical'' primitive is one whose security experiment only queries the underlying primitive a bounded polynomial number of times. We note that certain primitives such as quantum one-wayness or many-time PRSs technically allow an arbitrarily polynomial number of queries, but we can also consider bounded-query versions which are ``typical'' by this definition, and will be sufficient for our purposes.} primitive $A$, $A$ exists in a slightly modified CHRS model if and only if it exists in a unitary oracle model. In particular, separations that exist in the modified CHRS model also hold in a unitary oracle model.
\end{theorem}
Here, the CHRS variant provides all parties copies of the state $\ket{\psi-}:=\ket{\psi}-\ket{0}$ for a Haar-random $\ket{\psi}$, rather than $\ket{\psi}$ itself. 

It may seem that $\ket{\psi}$ and $\ket{\psi-}$ should have roughly the same power, and in one direction we show that this is true:
\begin{theorem}[informal]\label{thm:main_inf3} For any primitive $A$ whose security is described by a ``single-stage game'',\footnote{A single-stage game is a potentially interactive game between a \emph{single} adversary and challenger, as opposed to ``multi-stage'' games which have multiple isolated adversaries.} if $A$ exists in the plain CHRS model, it also exists in a $\ket{\psi-}$ model and hence in a unitary model.
\end{theorem}
On the other hand, we show that a generic solution to the converse is actually false, in the sense that there is no way to generically construct $\ket{\psi-}$ from even many copies of $\ket{\psi}$. However, the models are close enough that we expect typical separations in the plain CHRS model to also hold in the modified CHRS model. In fact, we show that the separation between commitments and many-time PRSs (or even quantum one-wayness) holds in the $\ket{\psi-}$ model. This reduces to showing that many-time PRSs and quantum one-wayness do not exist in the $\ket{\psi-}$ model, which we show through a very similar proof structure to the $\ket{\psi}$ case analyzed in~\cite{ARXIV:CheColSat24,EPRINT:BosCheNeh24}. Combined with Theorem~\ref{thm:main_inf} this gives an alternative proof of some of the main results of~\cite{ARXIV:CheColSat24,EPRINT:BosCheNeh24,EPRINT:BMMMY24} that is more modular and arguably conceptually simpler. Along the way, we identify significant conceptual errors in both~\cite{ARXIV:CheColSat24} and~\cite{EPRINT:BosCheNeh24}, rendering their unitary oracle separations incorrect. In particular, these bugs left separating commitments and one-wayness open relative to a unitary oracle, which we resolve.\footnote{Separating commitments and PRSs also appeared in~\cite{EPRINT:BMMMY24}, but there is no bug in their proof.}

We therefore argue that future work in this space should focus on proving separations in the $\ket{\psi-}$ model, which still has much of the simplicity of the plain CHRS model, but which is equivalent to a unitary separation through our Theorem~\ref{thm:main_inf}.

We also give conditions under which a converse of Theorem~\ref{thm:main_inf3} holds, meaning that separation in the $\ket{\psi}$ model implies an separation in the $\ket{\psi-}$ model and hence unitary model. In particular, we show the following:
\begin{theorem}[informal]\label{thm:main_inf2} For any primitive $A$ whose security and correctness experiments are defined by ``LOCC games'', $A$ exists in the plain CHRS model if and only if it exists in our $\ket{\psi-}$ model.
\end{theorem}
Here, LOCC (local operations classical communication) means that the game may involve many instances of the cryptosystem, but that those instances can only interact with each other via classical communication.

As a concrete application, we apply Theorem~\ref{thm:main_inf2} to the separation between quantum commitments and key agreement with classical communication, lifting the CHRS model separation from~\cite{TCC:AnaGulLin24} to the $\ket{\psi-}$ model and hence to a unitary model. This separation is new to our work.

\subsection{Motivation}

\paragraph{Black-box separations.} A fundamental goal in cryptography is to understand the relationship between various concepts. A positive result uses a primitive $A$ to build a primitive $B$. Meanwhile, a negative result, also called a separation, aims to shows that it is impossible to build $B$ from $A$. Stating negative results, however, require care: we typically believe both $A$ and $B$ exist, so a trivial way to ``build $B$ from $A$'' is to simply ignore $A$ and use the assumed solution to $B$ instead. Starting with the work of Impagliazzo and Rudich~\cite{STOC:ImpRud89}, the now-standard way to argue separations to provide a black-box oracle relative to which $A$ exists and $B$ does does not. Such oracle separations lose some generality, but they rule out all black-box constructions of $B$ from $A$, which constitute the vast majority of cryptographic constructions. Such separations are therefore called ``black-box separations.''

Often, a complete black-box separation is unknown. In this case, a common tactic is to further restrict the construction of $B$ in some way -- maybe it can only query the oracle at certain times, or it can only use the primitive $A$ in certain ways. There are numerous examples where such restrictions are made (e.g.~\cite{FOCS:GerMalRei01,TCC:GetMalMye07,EC:RotSegSha20,C:RotSeg20}). Such separations only rule out a restricted class of constructions, and are therefore considered much weaker evidence of an actually impossibility of building $B$ from $A$. For this reason, it is important to ensure that the restrictions are natural and reflect ``real world'' use cases as much as possible, lest the restricted separation be meaningless. Even though restricted separations are weaker, they may still be a starting point for a full separation, and can still be very useful for guiding protocol design. 

\paragraph{Quantum black-box separations.} These same questions remain fundamentally important as cryptography transitions to a quantum world. One of the promises of quantum cryptography is that it may be possible to overcome classical impossibilities. A notable example is multiparty computation (MPC): classical black-box separations indicate that MPC requires public key tools~\cite{STOC:ImpRud89,TCC:MahMajPra14}, whereas quantum (but black-box!) protocols give MPC from one-way functions~\cite{EC:GLSV21,C:BCKM21b}, the most basic classical symmetric key primitive.

In fact, quantum cryptography opens up an entire world of cryptographic protocols utilizing tools that are ``below'' one-way functions, often referred to as MicroCrypt. That is, it is now widely believed that there are (inherently quantum) primitives which do \emph{not} imply even the lowly one-way function, as evidenced by quantum black-box separations~\cite{TQC:Kretschmer21,STOC:KQST23}. For example, while one-way functions are unlikely to imply MPC classically, the \emph{opposite} is now believed quantumly: one-way functions not only imply quantum MPC, but quantum MPC is unlikely to imply one-way functions! Numerous other primitives are believed to exist without one-way functions, such as pseudorandom states and unitaries~\cite{C:JiLiuSon18}, one-way state generators~\cite{TQC:MorYam24}, and more. 

A central question is then to develop a new understanding of the relationship between quantum protocols, especially those existing below one-way functions. Recent works have begun to address this question, but as the bugs in~\cite{ARXIV:CheColSat24,EPRINT:BosCheNeh24} show, the process is challenging and error prone. We believe our results will be useful in this endeavor, as they provide general lifting results to facilitate moving from much easier CHRS (and similar) models into full unitary separations.

\subsection{The Prior Results}

Here, we discuss the unitary separations of~\cite{ARXIV:CheColSat24,EPRINT:BosCheNeh24,EPRINT:BMMMY24}. These works accomplish somewhat different results, but we focus on the common thread showing that commitments do not imply many-time PRSs in the unitary model. We note that these works do not necessary present their ideas as we do here, but our attempt here is to unify the underlying ideas into a common framework.

All of the results start from the existence of commitments and non-existence of many-time PRSs in the CHRS model as proved in~\cite{ARXIV:CheColSat24}, then translate those statements to the existence of commitments and non-existence of PRSs relative to a unitary oracle. The works all use some version of the unitary oracle $S_{\ket{\psi}}$ which exchanges the states $\ket{0}$ and $\ket{\psi}$, leaving all other states unaffected.\footnote{\cite{EPRINT:BMMMY24} use this oracle, but for a distribution over $\ket{\psi}$ that is not Haar-random. This distinction will not be important for this discussion.} We will call $S_{\ket{\psi}}$ the ``SWAP oracle'' and this model the ``SWAP model.''

The first step in these works is to notice that $S_{\ket{\psi}}$ readily allows for constructing the state $\ket{\psi}$. By simulating the states used by the CHRS commitments in this way, they obtain a SWAP-model commitment that is \emph{correct}. Arguing security works by contradiction: assume there is an adversary $A$ for the derived SWAP-model commitment, and compile it into an adversary $B$ for the original CHRS commitment. This is accomplished by simulating $S_{\ket{\psi}}$ (the SWAP oracle) given copies of $\ket{\psi}$ (the CHRS model). This in turn relies on two key ideas:
\begin{itemize}
    \item {\bf Simulating reflections about a state.} Consider the oracle $R_{\ket{\psi}}=I - 2\ketbra{\psi}$, which reflects around $\ket{\psi}$. This oracle allows for determining if an input state is equalto $\ket{\psi}$. In~\cite{C:JiLiuSon18}, is it is shown how to approximately simulate $R_{\ket{\psi}}$ using just several copies of $\ket{\psi}$. The error can be made an arbitrarily-small inverse polynomial by using an appropriate number of copies.
    \item {\bf Swapping states.} $S_{\ket{\psi}}$ can be seen as providing $\ket{\psi}$ when given the input state $\ket{0}$, and taking back the state when given $\ket{\psi}$. If we have a pool of states $\ket{\psi}$, we can simply give out and take back states from our pool as needed. Note that this requires some care, as queries to $S_{\ket{\psi}}$ can be in superpositions. Moreover, we need a way to tell if the input state is actually $\ket{\psi}$; fortunately, this is accomplished with $R_{\ket{\psi}}$, which is simulated as above.
\end{itemize}

This strategy for simulating $S_{\ket{\psi}}$ first appeared in~\cite{ITCS:Zhandry24b} in an entirely different context. There, it was proved that for a Haar-random state $\ket{\psi}$, the simulation is indistinguishable from $S_{\ket{\psi}}$ to arbitrarily-small inverse polynomial error. This result was used as a black-box by~\cite{EPRINT:BosCheNeh24}, and very similar variants of it were developed in~\cite{ARXIV:CheColSat24,EPRINT:BMMMY24}. Through this simulation, the adversary $A$ is successfully compiled into an adversary $B$ in the CHRS model, reaching a contradiction. Thus commitments exist in the SWAP model.

The next step is to argue the non-existence of PRSs in the SWAP model. This can be seen as a dual to the case above, with significant caveats. Namely, we now assume toward contradiction that there is a protocol $C$ for PRSs in the SWAP model and derive a PRS $D$ in the CHRS model, reaching a contradiction. To prove that $D$ is a PRS we need in particular to prove that it is secure. For this, we suppose toward contradiction that there is an adversary for $D$, and lift it into an adversary for $C$, which cannot exist by assumption. At first glance, this appears identical to the commitment case above, except we have exchanged the roles of adversary and protocol.

\paragraph{A Problem.} Following this intuition, we may hope to use such simulation to prove the impossibility of PRSs in the SWAP model. Let $D$ be $C$ but where queries to $S_{\ket{\psi}}$ are simulated using copies of $\ket{\psi}$. Then for any hypothetical adversary for $D$ (which is in the CHRS model and therefore gets copies of $\ket{\psi}$) we turn it into an adversary for $C$ in the SWAP model by using $S_{\ket{\psi}}$ to generate copies of $\ket{\psi}$.

However, there is a major issue with this claim: the simulation of $S_{\ket{\psi}}$ uses state -- namely, the several copies of $\ket{\psi}$ provided -- and this state becomes \emph{entangled} with the ultimate output. This is because, for example, a query to $S_{\ket{\psi}}$ on $\ket{0}+\ket{1}$ gives $\ket{\psi}+\ket{1}$. Using the simulation, the output is in superposition of having swapped out 0 or 1 copies of $\ket{\psi}$. The simulation guarantee is only that the view of any adversary given query access to $S_{\ket{\psi}}$ is indistinguishable from the simulated oracle, not that the two cases are literally the same. 

Fortunately, this is not a problem for showing the security of commitments in the SWAP model. The entire security experiment acts as a distinguisher for the simulation, and the simulation guarantees that the output of the experiment has approximately the same distribution.

The problem instead comes in showing that PRSs do not exist in the SWAP model. The PRS state is now constructed via the simulated oracle $S_{\ket{\psi}}$; the simulator for this oracle now becomes entangled with the final PRS state, meaning the state is not pure. This breaks even the correctness of PRSs, as PRSs inherently need to have pure-state outputs to be non-trivial. But it gets even worse, as the adversary sees many copies of the PRS state. An adversary can always test for purity by applying the swap test to its several copies. In a true PRS the swap test will accept, but now in our simulated PRS the swap test will fail. Thus, the derived CHRS-model PRS will not be correct \emph{nor} secure.

\begin{remark}One may wonder where the simulation indistinguishability guarantee fails. First, the guarantee only holds relative to efficient distinguishers, whereas the purity requirement of a PRS is statistical. This issue can be fixed by describing purity as a game, performing the swap test on two copies as in our attack above. But here there is a more subtle issue: the state of each copy is simulated using a separate instance of the simulator, but the indistinguishability of the simulation only works when there is a single simulator for all queries. There is even a simple distinguisher that can distinguish the case where there are two independent simulators. Thus, even though the purity requirement is turned into an efficient game, the simulation guarantee does not apply. 
\end{remark}

\noindent The three prior works address this problem in different ways:
\begin{itemize}
    \item \cite{ARXIV:CheColSat24} restricts to primitives that have no correctness requirements. The authors then claim that quantum pseudorandomness and EFI pairs (which are equivalent to commitments) are examples of such correctness-less primitives. However, this is simply not true -- PRSs have the purity requirement, and EFI pairs require statistical far-ness. This gap means their proof is actually currently incorrect. The problem is even worse since, as we observed above, the simulation issue actually affects the security as well.  The issue has been confirmed by the authors, who are currently working on a fix.
    \item In \cite{EPRINT:BosCheNeh24}, rather than trying to turn a SWAP-model PRS into a CHRS-model PRS, they give a direct attack on any SWAP-model PRS. Their attack, however, still starts with a generic CHRS-model attack and lifts it to the SWAP-model through simulation. This avoids the issue of~\cite{ARXIV:CheColSat24} since they do not need to prove that a simulated PRS is correct. However, their proof actually still contains a major but subtle bug due to incorrect usage of the simulator's guarantees. The issue has been confirmed by the authors, who are retracting their claims of a unitary-oracle separation.
    \item \cite{EPRINT:BMMMY24} likewise employ a direct attack on PRSs in the SWAP model, which requires reproducing the CHRS impossibility with modifications to work in the SWAP model. However, their approach is quite different from~\cite{EPRINT:BosCheNeh24}, and appears correct.
\end{itemize}
We note that the bugs in~\cite{ARXIV:CheColSat24} and~\cite{EPRINT:BosCheNeh24} leave the question of separating commitments from quantum one-wayness relative to a unitary oracle open.

\subsection{Our Results}

\paragraph{A general lifting theorem via indifferentiability.} Our aim is to abstract the techniques of~\cite{ARXIV:CheColSat24,EPRINT:BosCheNeh24,EPRINT:BMMMY24} into a general theorem that lifts separations from the CHRS model to separations in a unitary model. We do so via the notion of indifferentiability~\cite{TCC:MauRenHol04}. Very roughly,~\cite{TCC:MauRenHol04} considers attempting to build one oracle $A$ from another $B$ via a construction $C^B$. The central observation isthat it is not enough to show the \emph{indistinguishability} of $C^B$ from $A$. This is because the adversary actually has access to the underlying $B$ oracle directly, and can use this access to potentially break the construction. \emph{Indifferentiability} resolves this issue, by giving the adversary access to both $C^B$ and $B$ itself. Concretely, indifferentiability, considers the following two worlds:
\begin{itemize}
    \item {\bf The Real World:} The adversary can query $C^B$ and $B$
    \item {\bf The Ideal World:} There is a simulator $S^A$ that gets queries to $A$, and tries to simulate a $B$ that is consistent with $A$. Here, the adversary's oracles $(C^B,B)$ are replaced by $(A,S^A)$.
\end{itemize}
Indifferentiability asks that, for any possible distinguisher, there exists a simulator $S$ that makes the Real and Ideal worlds indistinguishable.~\cite{TCC:MauRenHol04} moreover demonstrate a composition theorem, showing that indifferentiability implies that any primitive whose security is described by a ``single-stage game'' which exists relative to $A$ also exists relative to $B$. Setting $A=C^B$ gives the construction, and the simulator $S$ is used to prove security by translating an adversary relative to $B$ into an adversary relative to $A$. Single-stage means that the security experiment involves one adversary interacting with a challenger, as opposed to ``multi-stage'' games which have multiple isolated adversaries. Fortunately, almost all cryptographic games are single-stage. A stronger notion of \emph{reset} indifferentiability was proposed in~\cite{EC:RisShaShr11}, which allows the composition theorem to apply even for multi-stage games. While these composition theorem were first demonstrated classically, they extend to quantum protocols/adversaries as well.

\medskip

Using indifferentiability, we show that the SWAP model $S_{\ket{\psi}}$ is actually \emph{equivalent} to a slightly modified CHRS variant where all parties are given access to $\ket{\psi-}:=\ket{\psi}-|0\rangle$. That is, any construction or adversary can be translated between the two models. Note that we actually do not need $\ket{\psi}$ to be Haar random, and we can use arbitrary states such as the subset states of~\cite{EPRINT:BMMMY24}.

The proof of this builds on the simulation idea using techniques from~\cite{EPRINT:BMMMY24}. The proof of~\cite{EPRINT:BMMMY24} shows that the $\ket{\psi-}$ states can actually be used to \emph{statelessly} simulate $S$. The idea is to change our view of $S_{\ket{\psi}}$: swapping between $\ket{0}$ and $\ket{\psi}$ is equivalent to just negating the phase on $\ket{\psi-}$. Given several copies of $\ket{\psi-}$, we can simulate this phase change through a swap test, without any need to explicitly swap copies of $\ket{\psi}$ in and out of our pool. Note that while this simulator technically still has state -- namely the pool of $\ket{\psi-}$ states -- up to inverse-polynomial error it actually will not become entangled with the algorithm querying $S_{\ket{\psi}}$. In particular, it is fine for different algorithms to use different copies of $\ket{\psi-}$ to simulate, unlike the simulation using $\ket{\psi}$.

Abstracting this idea, we show that the simulation is fully general. Using the trivial observation that $S_{\ket{\psi}}$ can simulate $\ket{\psi-}$ by querying on $\ket{-}$, we therefore establish the full equivalence of the $S_{\ket{\psi}}$ and $\ket{\psi-}$ models. Since our simulators are stateless, we acheive even the stronger notion of reset-indifferentiability, meaning our results apply to even multi-stage games. This gives Theorem~\ref{thm:main_inf}.

\begin{remark}One caveat is that simulation still incurs inverse polynomial error. This is fine for the \emph{adversary}, since we can make the error sufficiently small (by using more copies of $\ket{\psi-}$) so that the error is less than the adversary's advantage. But for simulating the \emph{construction}, we have to use a bounded polynomial number of copies, resulting in an arbitrarily-small but non-negligible correctness error. This further impacts security, as each call to the cryptographic primitive in the security experiment results in a correctness error, which could turn into adversarial advantage. Fortunately, the prior impossibilities are sufficiently robust that they straightforwardly hold even with these inverse-polynomial errors.\footnote{A more subtle caveat is that for many-time primitives such as many-time PRSs or quantum one-wayness, the primitive is called an unbounded polynomial number of times in the security experiment, while the error in each invocation in a bounded inverse polynomial. Thus, by setting the number of queries of the experiment high enough, the errors from each query will compound, giving an overall error that could approach 1. However, this will not typically be a problem, as the impossibilities we deal with give attacks that use a concrete polynomial number of queries, thus ruling out even bounded versions of the primitives in question.}\end{remark}

\paragraph{Separations in the $\ket{\psi-}$ model.} While $\ket{\psi-}$ appears very close to $\ket{\psi}$, we technically arrive at a different common-state model. This means we cannot use the separations in the CHRS model as a black box without inspecting the proof. Nevertheless, we show that the impossibilities of PRSs (and even one-way state generators) from~\cite{ARXIV:CheColSat24,EPRINT:BosCheNeh24,EPRINT:BMMMY24} readily translate to the $\ket{\psi-}$ model with minimal modifications. Thus, we separate commitments from many-time PRSs and even one-way state generatorss in a unitary model, in a way that is much simpler and more modular than the prior work on unitary oracle separations~\cite{EPRINT:BMMMY24} and also circumvents the bugs in~\cite{ARXIV:CheColSat24,EPRINT:BosCheNeh24}. 

\paragraph{Lifting $\ket{\psi}$ to $\ket{\psi-}$.} It is natural to wonder if the $\ket{\psi}$ and $\ket{\psi-}$ models are equivalent as well. Combined with our Theorem~\ref{thm:main_inf}, this would allow for immediately translating $\ket{\psi}$ separations into unitary separations without having to inspect the underlying $\ket{\psi}$ separations. Unfortunately, we explain that a stronger version of this hope is not possible: it is impossible to simulate $\ket{\psi-}$ from $\ket{\psi}$, where both use the \emph{same} Haar random state $\ket{\psi}$. Intuitively, this makes sense: Haar-random states cannot be copied nor deleted, given just a polynomial-number of copies. Therefore, given any polynomial-number of copies of $\ket{\psi}$, one will always have an integer number of them. On the other hand $\ket{\psi-}$ is essentially half of a $\ket{\psi}$ state, which is not an integer.

The good news is that we give a general condition under which it \emph{is} possible to have $\ket{\psi}$ and $\ket{\psi-}$ be equivalent. For starters, any construction in the $\ket{\psi}$ model can be lifted to the $\ket{\psi-}$ model, assuming the security experiment is single-stage. Obtaining a correct construction requires obtaining copies of $\ket{\psi}$ from copies of $\ket{\psi-}$, which can be accomplished through post-selecting $\ket{\psi-}$ on the state not being 0. In order for the derived construction to be \emph{secure}, we then need to simulate an adversary relative to $\ket{\psi-}$ using only copies of $\ket{\psi}$. Here, since we assume the security experiment is single-stage, the stateful simulation from above actually suffices. This gives Theorem~\ref{thm:main_inf3}. This gives half of the goal of lifting a separation to the unitary model, namely the part showing that primitives (e.g. commitments) exist in the unitary model.

Then we consider restrictions that yield give a converse to Theorem~\ref{thm:main_inf3}, which gives the other half of a separation by showing that other primitives do not exist in the unitary model. In particular, we consider ``LOCC'' games as games where there may be several invocations of the underlying cryptosystems, but where the different invocations cannot talk to each other except via classical communication. We show, as long as we restrict to primitives where security \emph{and} correctness are defined by LOCC games, that $\ket{\psi}$ and $\ket{\psi-}$ are in fact equivalent. This equivalence uses the stateful simulator for $S_{\ket{\psi}}$ to give $\ket{\psi-}$, but we give a much stronger simulation guarantee than that of prior works: you actually \emph{can} have multiple instances of the simulator, as long as they cannot talk to each other except for LOCC communication. This gives Theorem~\ref{thm:main_inf2}.

We apply this to the separation of~\cite{TCC:AnaGulLin24}, who show that key agreement with classical communication does not exist in the CHRS model, thereby separating it from commitments. The security and correctness games for such key agreement are LOCC, and thus our result shows we can lift it into a separation in the $\ket{\psi-}$ model, and hence via Theorem~\ref{thm:main_inf} into a separation in a unitary model.

\section{Preliminaries}
\subsection{Cryptographic Preliminaries}

\newcommand{\Prim}{Prim}

\begin{definition}
    An idealized model $\mathcal{P} = \{\mathcal{P}_n\}_{n \in \N}$ is a family of distributions over isometries. A uniform/non-uniform oracle algorithm $\A^{\mathcal{P}}$ is a QPT algorithm with access to oracle gates performing some potentially stateful process $\cal{P}$.
\end{definition}

\begin{remark}
    Throughout this paper, we will consider all adversaries to have access to the idealized model only at the index corresponding to the security parameter. This is common in the indifferentiability setting~\cite{TCC:MauRenHol04}, and is .
    
    One could consider a more general model where adversaries can query the idealized model at any security parameter, as is considered in some of the prior work. It is not difficult to port all of our results over to this setting. However, since our constructions are only secure at sufficiently high security parameters, this would require applying a different "brute-force" construction for low security parameters. As these technical details are relatively involved, albeit not very informative, we omit such a detailed discussion in this work.
\end{remark}

\begin{definition}
    A cryptographic primitive $\Prim = (\Prim_1,\dots,\Prim_t)$ is a tuple of potentially stateful uniform QPT algorithms.

    A cryptographic adversary $\A = (\A_1,\dots,\A_\ell)$ is a tuple of potentially stateful \textit{non-uniform} QPT algorithms.

    If $\mathcal{P}$ is an idealized model, a cryptographic primitive $\Prim^{\mathcal{P}} = (\Prim_1^{\mathcal{P}}, \dots, \Prim_t^{\mathcal{P}})$ relative to $\mathcal{P}$ is a tuple of potentially stateful uniform QPT oracle algorithms. Like-wise, a cryptographic adversary $\A^{\mathcal{P}}=(\A_1^{\mathcal{P}},\dots,\A_\ell^{\mathcal{P}})$ is a tuple of potentially stateful non-uniform QPT oracle algorithms.
\end{definition}

\begin{definition}
    A cryptographic game $G$ consists of the following
    \begin{enumerate}
        \item The syntax of a primitive $\Prim$. This consists of the number of algorithms $\ell_{\Prim}$, the input and output space of each algorithm, and whether each procedure is stateful or not.
        \item A syntax of an adversary $\A$. This consists of the number of algorithms $\ell_{\A}$, the input and output space of each algorithm, and whether each procedure is stateful or not.
        \item A QPT oracle procedure which, on input security parameter $1^n$, makes queries to a primitive $\Prim(1^n, \cdot)$ and an adversary $\A(1^n,\cdot)$ satisfying the above syntax. At the end, the procedure should produce a bit $b$.
    \end{enumerate}
    We say a primitive $\Prim$ (like-wise adversary $\A$) is compatible if $\Prim = (\Prim_1,\dots,\Prim_{\ell_{\Prim}})$, the input and output spaces of each $\Prim_i$ match those described by $G$, and whether or not each $\Prim_i$ is stateful matches with the syntax described by $G$.
    
    We will denote the interaction between the game $G$ and a compatible primitive $\Prim$ and adversary $\A$ as $G(1^n,\Prim,\A)$. 

    If the syntax of $G$ specifies that $\A$ consists of a single, stateful procedure, then we say that $G$ is single-stage.
    
    We say that $\Prim$ is $(t,\epsilon;c)$-secure under $G$ if for all quantum adversaries $\A$ running in time $t$, for all sufficiently large $n$,
    $$\Pr[G(1^n,\Prim,\A) \to 1] - c \leq \epsilon(n)$$

    If $\Prim^{\mathcal{P}}$ is a cryptographic primitive relative to $\mathcal{P}$, then we say that $\Prim^{\mathcal{P}}$ is $(t,\epsilon;c)$-secure under $G$ (relative to $\mathcal{P}$) if for all oracle adveraries $\A^{\mathcal{P}}$ making at most $t$ oracle queries, for all sufficiently large $n$,
    $$\Pr[G(1^n,\Prim^{\mathcal{P}},\A^{\mathcal{P}_n}) \to 1] - c \leq \epsilon(n)$$

    We say that $\Prim$ is $(\epsilon;c)$-secure under $G$ if it is $(t,\epsilon;c)$-secure for all $t \leq \poly(n)$.

    We say that $\Prim$ is $c$-secure under $G$ if it is $(n^{-d};c)$-secure for all $d\in \N$.

    We will omit $c$ whenever $c=0$.
\end{definition}

\begin{remark}
    Note that here we allow adversaries in an idealized model to be inefficient, as long as they are query bounded. It would also be reasonable to consider time and query bounded adversaries. In this setting most separations can also be achieved by providing access to a sufficiently strong oracle independent of the idealized model (for example, $\mathbf{UnitaryPSPACE}$ in~\cite{EPRINT:BosCheNeh24}).
\end{remark}
\subsection{Local Operations Classical Communication}

\begin{definition}
    Let $\mathcal{O} = (\mathcal{O}_1,\dots, \mathcal{O}_t)$ be a tuple of oracles. An oracle LOCC algorithm $\mathcal{A}^\mathcal{O} = (\mathcal{A}_1^{\mathcal{O}_1}, \dots, \mathcal{A}_t^{\mathcal{O}_t})$ is a tuple of quantum interactive oracle algorithms taking in classical inputs and outputs. We write $\mathcal{A}^{\mathcal{O}} \to z$ to denote the process
    $$\mathcal{A}_1^{\mathcal{O}_1} \leftrightarrows \dots \leftrightarrows \mathcal{A}_t^{\mathcal{O}_t}$$
    followed by returning $\mathcal{A}_1^{\mathcal{O}_1}$'s output $z$.

    We say that a $\A^\mathcal{O}$ makes at most $T$ queries to its oracle if the total number of queries made by $\mathcal{A}_1^{\mathcal{O}_1}, \dots, \mathcal{A}_t^{\mathcal{O}_t}$ to the oracles $\mathcal{O}_1,\dots,\mathcal{O}_t$ is at most $T$. We will also work with LOCC algorithms over pairs of oracles $\mathcal{O}^1, \mathcal{O}^2 = ((\mathcal{O}_1^1,\mathcal{O}_1^2),\dots, (\mathcal{O}_t^1,\mathcal{O}_t^2))$, in which case we say that an algorithm $\A^{\mathcal{O}^1,\mathcal{O}^2}$ makes at most $T_1$ queries to its first oracle if the total number of queries made by $\mathcal{A}_1^{\mathcal{O}_1^1,\mathcal{O}_1^2}, \dots, \mathcal{A}_t^{\mathcal{O}_t^1,\mathcal{O}_t^2}$ to the oracles $\mathcal{O}_1^1,\dots,\mathcal{O}_t^1$ is at most $T$. The statement $\A^{\mathcal{O}^1,\mathcal{O}^2}$ makes at most $T_2$ queries to its second oracle is defined analogously.
\end{definition}

\begin{definition}
    We say that a pair of tuples of oracles $\mathcal{O} = (\mathcal{O}_1,\dots, \mathcal{O}_t)$ and $\mathcal{O}' = (\mathcal{O}_1',\dots, \mathcal{O}_t')$ are $(T,\epsilon)$-LOCC indistinguishable if for all oracle LOCC algorithms $\mathcal{A}^{\cdot}$ making at most $T$ oracle queries,
    $$\abs{\Pr[\mathcal{A}^{\mathcal{O}} \to 1] - \Pr[\mathcal{A}^{\mathcal{O}'} \to 1]} \leq \epsilon$$

    We say that a pair of states $\rho_{A_1,\dots,A_t},\sigma_{A_1,\dots,A_t}$ are $(T,\epsilon)$-LOCC indistinguishable if the oracles $(\O^\rho_1,\dots,\O^\rho_t),(\O^\sigma_1,\dots,\O^\sigma_t)$ are $(T,\epsilon)$-LOCC indistinguishable where $\O^{\rho}_i$ returns register $A_i$ of $\rho$ and $\O^{\sigma}_i$ returns register $A_i$ of $\sigma$.

    We say that a pair of oracles/states are $\epsilon$-LOCC indistinguishable if they are $(\infty,\epsilon)$-LOCC indistinguishable.
\end{definition}

\begin{definition}[LOCC cryptographic games]
    We say that a cryptographic game $G$ is LOCC if it can be represented by a LOCC computation. In particular, we say that $G$ is LOCC if it is represented by the following computation
    \begin{enumerate}
        \item There exist some $G_1^{(\cdot)},\dots,G_{\ell}^{(\cdot)}$ such that $G(\Prim,\A,1^n)$ is represented by the computation
        $$G_1^{\Prim_1} \leftrightarrows \dots \leftrightarrows G_{\ell_{\Prim}}^{\Prim_{\ell_{\Prim}}}\leftrightarrows G_{\ell_{\Prim} + 1}^{\A_1} \leftrightarrows \dots \leftrightarrows G_{\ell_{\Prim} + \ell_{\A}}^{\A_{\ell_\A}}$$
        followed by repeating the output of $G_1$. Here, all messages sent between parties $G_i,G_j$ must be classical.
        \item We further require that for all stateless $\Prim_i$, $G_i^{\Prim_i}$ makes at most one query to its oracle. 
        \item Similarly, for all stateless $\A_i$, $G_{\ell_{\Prim}+i}^{\A_i}$ makes at most one query to its oracle.
    \end{enumerate}
\end{definition}
\subsection{Indifferentiability}

In this section, we define indifferentiability and relevant variants. We also state useful composition theorems which apply to indifferentiable constructions. Formal proofs of these composition theorems are provided for completeness in~\Cref{sec:comp}.

\begin{definition}[Standard indifferentiability]
    Let $\mathcal{P},\mathcal{Q}$ be two idealized primitives. We say that a construction $C^{\mathcal{P}}$ is $(T_{\Sim},T_1,T_2,\delta)$-indifferentiable from $\mathcal{Q}$ if for all adversaries $\A^{\O_1,\O_2}$ making at most $T_1(n)$ queries to $\O_1$ and $T_2(n)$ queries to $\O_2$, there exists a potentially stateful simulator $\Sim^{\mathcal{Q}}_n$ making at most $T_{\Sim}(n)$ queries to its oracle such that for all sufficiently large $n$,
    $$\abs{\Pr[\A^{\mathcal{P}_n,C^{\mathcal{P}_n}}(1^n) \to 1] - \Pr[\A^{\Sim^{\mathcal{Q}_n},\mathcal{Q}_n}(1^n) \to 1]} \leq \delta(n)$$

    If in addition the simulator is stateless, then we say that $C^{\cal{P}}$ is \textit{reset indifferentiable} from $\cal{Q}$.
\end{definition}

\begin{theorem}[Bounded query composition theorem]\label{thm:indiffcomp}
Let $\mathcal{P},\mathcal{Q}$ be two idealized primitives, and let $C^{\mathcal{P}}$ be a construction $(T_{\Sim},T_1,T_2,\delta)$-indifferentiable from $\mathcal{Q}$. Let $G$ be any single-stage cryptographic game making at most $T_{G,1}$ queries to its primitive and $T_{G,2}$ queries to its adversary, and let $\Prim^{\mathcal{Q}}$ be any primitive relative to $\mathcal{Q}$ making at most $T_{\Prim}$ queries to its oracle. Let $\epsilon,T_{\A}:\N\to [0,1]$ be any functions and let $c$ be any constant.

As long as $T_{G,1}\cdot T_{\Prim} \leq T_2$ and $T_{G,2}\cdot T_{\A} \leq T_1$, then if $\Prim^{\cal{Q}}$ is $(T_{\A}\cdot T_{\Sim},\epsilon;c)$-secure under $G$ relative to $\cal{Q}$, then $\Prim^{C^{\cal{P}}}$ is $\left(T_{\A}, \epsilon + \delta;c\right)$-secure under $G$ relative to $\cal{P}$.
\end{theorem}

\begin{corollary}[General composition theorem]\label{cor:compcor}
    Let $\mathcal{P},\mathcal{Q}$ be two idealized primitives, let $C^{\mathcal{P}}$ be a construction of $\mathcal{Q}$ from $\mathcal{P}$, and let $\delta:\N\to [0,1]$. If, for all $p = \poly(n)$, there exists some $q = \poly(n)$ such that $C^{\mathcal{P}}$ is $(q,p,p,\delta)$-indifferentiable from $\mathcal{Q}$, then the following holds:

    For all primitives $\Prim^{\mathcal{Q}}$ and single-stage cryptographic games $G$, if $\Prim^{\mathcal{Q}}$ is $(\epsilon;c)$-secure under $G$ relative to $\mathcal{Q}$, then $\Prim^{C^{\mathcal{P}}}$ is $(\epsilon+\delta;c)$-secure under $G$ relative to $\mathcal{P}$.
\end{corollary}

\begin{remark}
    If, in any of these theorems, the construction $C^{\mathcal{P}}$ additionally satisfies reset indifferentiability, then we can remove the requirement that $G$ is single-stage~\cite{EC:RisShaShr11}.
\end{remark}

\paragraph*{What does it mean for a stateful construction to be indifferentiable?}
For the composition lemma to hold, we need indifferentiability to hold when the adversary has query access to many different copies of the construction. This is essentially the reverse of reset indifferentiability. However, we may resign ourselves to a weaker composition theorem, which holds only for particular types of security games. In particular, we could consider only LOCC security games, where the game can be written as a LOCC protocol where each local operation only operates on one copy of the primitive. To handle this, we need indifferentiability to hold against multiple adversaries communicating over LOCC channels where each adversary has access to a fresh version of the construction. We will call such a property \textit{LOCC indifferentiability}.

\begin{definition}[LOCC indifferentiability]
    Let $\mathcal{P},\mathcal{Q}$ be two idealized primitives. We say that a construction $C^{\mathcal{P}}$ is $(\ell_\A,T_{\Sim},T_1,T_2,\epsilon)$-LOCC indifferentiable from $\mathcal{Q}$ if for all LOCC adversaries $\A^{(\cdot),(\cdot)} = (\A_1^{(\cdot),(\cdot)},\dots,\A_{\ell_\A}^{(\cdot),(\cdot)})$ making at most $T_1,T_2$ queries to its first and second oracles respectively, there exists a simulator $\Sim^{\mathcal{Q}}$ making at most $T_{\Sim}$ queries to its oracle such that 
    $$\abs{\Pr[\A^{((\mathcal{P},C_1^{\mathcal{P}}), \dots, (\mathcal{P},C_t^{\mathcal{P}}))} \to 1] - \Pr[\A^{((\Sim_1^{\mathcal{Q}},\mathcal{Q}), \dots, (\Sim_5^{\mathcal{Q}},\mathcal{Q}))} \to 1]} \leq \epsilon$$
    where $C_i$, $\Sim_i$ are fresh instantiations of $C$ and $\Sim$ for all $i\in [k]$.
\end{definition}

\begin{theorem}[LOCC composition theorem]\label{thm:locccomp}
Let $\mathcal{P},\mathcal{Q}$ be two idealized primitives, and let $C^{\mathcal{P}}$ be a (possibly stateful) construction $(\ell_\A,T_{\Sim},T_1,T_2,\delta)$-LOCC indifferentiable from $\mathcal{Q}$. Let $G$ be any LOCC cryptographic game making at most $T_{G,1}$ queries to its primitive and $T_{G,2}$ queries to its adversary, with at most $\ell_\A$ adversaries. Let $\Prim^{\mathcal{Q}}$ be any primitive relative to $\mathcal{Q}$ making at most $T_{\Prim}$ queries to its oracle. Let $\epsilon,T_{\A}:\N\to [0,1]$ be any functions and let $c$ be any constant.

As long as $T_{G,1}\cdot T_{\Prim} \leq T_2$ and $T_{G,2}\cdot T_{\A} \leq T_1$, then if $\Prim^{\cal{Q}}$ is $(T_{\A}\cdot T_{\Sim},\epsilon;c)$-secure under $G$ relative to $\cal{Q}$, then $\Prim^{C^{\cal{P}}}$ is $\left(T_{\A}, \epsilon + \delta;c\right)$-secure under $G$ relative to $\cal{P}$.
\end{theorem}
\section{Our Idealized Models}

\begin{notation}
    For any state $\ket{\phi}$, we define the state
    $$\ket{\phi-} \coloneq \frac{1}{\sqrt{2}}\left(\ket{0} - \ket{\phi}\right)$$
\end{notation}

\begin{definition}
    Let $\D = \{\D_n\}_{n\in \N}$ be any family of distributions over pure states over $n$ qubits (i.e. elements of $\C^{2^n}$). We define a series of oracles parameterized by $\D$. For each of these oracles, let $\ket{\phi_n} \gets \D$ be a state sampled during initialization, before any party receives oracle access. We will also consider the space $\C^{2^n+1}$, which will be the space spanned by the output space $\C^{2^n}$ of $\D_n$ and an orthogonal basis vector $\ket{0}$.
    \begin{enumerate}
        \item We define the structured state model (SSM) over $\D$ to be the oracle $\SSM_{n}:\C \to \C^{2^n}$ defined by the isometry
            $$\ket{0} \mapsto \ket{\phi_n}$$
        \item We define the weighted structured state model (SSM-) over $\D$ to be the oracle $\SSMM_{n}:\C \to \C^{2^n + 1}$ defined by the isometry
            $$\ket{0} \mapsto \ket{\phi_n-} = \frac{1}{\sqrt{2}}\left(\ket{0} - \ket{\phi_n}\right)$$
        \item We define the swap structured state model (Swap) over $\D$ to be the oracle $\Swap_n:\C^{2^{n} + 1} \to \C^{2^{n} + 1}$ defined by the unitary which swaps $\ket{0}$ and $\ket{\phi_n}$ and acts as identity everywhere else. Formally,
        $$\Swap_n = I - \ketbra{0} - \ketbra{\phi_n} + \ketbra{0}{\phi_n} + \ketbra{1}{0}$$
    \end{enumerate}
\end{definition}

\begin{remark}
    Note that in these definitions we assume that $\ket{0}$ is some state orthogonal to $\ket{\phi}$. We can explicitly force this to happen by setting $\ket{0}=\ket{0}\ket{0}$ and replacing $\ket{\phi_n}$ by $\ket{1}\ket{\phi_n}$.
\end{remark}
\section{Building $\SSM$ from $\SSMM$}

\begin{definition}
    We say a distributions over states $\D$ is balanced if it does not change when we apply a random phase. That is, if we define $\D'$ to do the following
    \begin{enumerate}
        \item Sample $\ket{\phi} \gets \D_n$
        \item Sample $\theta \gets [0,2\pi]$
        \item Output $e^{i\theta} \ket{\phi}$
    \end{enumerate}
    then $\D = \D'$.

    Similarly, we say a family of distributions $\D = \{\D_n\}_{n \in \N}$ is balanced if $\D_n$ is balanced for all $n\in \N$.
\end{definition}

\begin{definition}
    For any $t \in \N$, set $S \subseteq [t]$, state $\ket{\phi} \in \C^{2^n}$, define
    $$\ket{\Set{t,S,\phi}} \coloneq \bigotimes_{i = 1}^{t} (\Ind_{i \notin S}\ket{0} - \Ind_{i \in S}\ket{\phi})$$
    For any $t \in \N$, $0 \leq c \leq t$, state $\ket{\phi} \in \C^{2^n}$, define
    $$\ket{\Rep{t,c,\phi}} \coloneq {t \choose c}^{-1/2}\sum_{S\subseteq [t]:|S| = c}\ket{\Set{t,S,\phi}}$$
    We will sometimes omit $t$ when clear from context.
\end{definition}

\begin{definition}
    $B(t,p)$ is the binomial distribution over $t$ instances. Formally, it is the distribution defined by
    $$\Pr[B(t,p) \to c] = {t \choose c}\cdot p^c\cdot (1-p)^{t-c}$$
    for all $0 \leq c \leq t$.
\end{definition}

\begin{lemma}\label{lem:binom}
    Let $\D$ be any a balanced distribution and let $t_1,t_2\in \N$. 
    Define
        $$\rho_n = \E_{\ket{\phi}\gets \D}\left[ \ketbra{\phi-}^{\otimes t_1}\otimes \ketbra{\phi}^{\otimes t_2} \right]$$
        $$\rho_n' = \E_{\ket{\phi} \gets \D,c\gets B(t_1,1/2)}\left[\ketbra{\Rep{t_1,c,\phi}}\otimes \ketbra{\phi}^{\otimes t_2}\right]$$
    Then $\rho = \rho'$.
\end{lemma}

\begin{proof}
    In particular, we will show that $\rho=\rho'$ as density matrices. Note that
    \begin{equation}\label{eq:phimintorep}
        \begin{split}
            \ket{\phi-}^{\otimes t_1} = \frac{1}{\sqrt{2^{t_1}}}\sum_{S \subseteq [t_1]} \bigotimes_{i = 1}^{t_1}(\Ind_{i \notin S}\ket{0} - \Ind_{i \in S}\ket{\phi})\\
            = \frac{1}{\sqrt{2^{t_1}}}\sum_{S \subseteq [t_1]} \ket{\Set{S,\phi}}= \sum_{c = 1}^{t_1} \frac{{t_1 \choose c}^{1/2}}{\sqrt{2^{t_1}}} \ket{\Rep{c,\phi}}
        \end{split}
    \end{equation}
    And so
    \begin{equation}\label{eq:help0}
        \begin{split}
            \rho = \E_{\ket{\phi}\gets \mu_n}\left[ \ketbra{\phi-}^{\otimes t_1}\otimes \ketbra{\phi}^{\otimes t_2} \right]\\
            = \E_{\ket{\phi}\gets \mu_n}\left[ \sum_{c,c' \in [t_1]} \frac{{t_1 \choose c}^{1/2} {t_1 \choose c'}^{1/2}}{2^{t_1}} \op{\Rep{c,\phi}}{\Rep{c',\phi}}\otimes \ketbra{\phi}^{\otimes t_2} \right]\\
            = \E_{\ket{\phi}\gets \mu_n}\left[ \sum_{c \in [t_1]} \frac{{t_1 \choose c}}{2^{t_1}} \op{\Rep{c,\phi}}{\Rep{c,\phi}}\otimes \ketbra{\phi}^{\otimes t_2} \right]\\ + \sum_{c \neq c' \in [t_1]} \frac{{t_1 \choose c}^{1/2} {t_1 \choose c'}^{1/2}}{2^{t_1}} \E_{\ket{\phi}\gets \mu_n} \left[\op{\Rep{c,\phi}}{\Rep{c',\phi}}\otimes \ketbra{\phi}^{\otimes t_2}\right]\\
            = \rho' + \frac{1}{2^{t_1}}\sum_{S,S'\subseteq [t_1],|S|\neq |S'|}\E_{\ket{\phi}\gets \mu_n}[\op{\Set{S,\phi}}{\Set{S',\phi}}\otimes \ketbra{\phi}^{\otimes t_2}]
        \end{split}
    \end{equation}
    We then claim that since $\D$ is balanced, for all $S,S' \subseteq [t_1]$ such that $|S|\neq |S'|$,
    \begin{equation}
        \E_{\ket{\phi}\gets \mu_n} \left[\op{\Set{S,\phi}}{\Set{S',\phi}}\otimes \ketbra{\phi}^{\otimes t_2}\right] = 0
    \end{equation} for all $c \neq c'$.
    In particular, since $\D$ is balanced, we have
    \begin{equation}\label{eq:help1}
        \begin{split}
            \E_{\ket{\phi}\gets \mu_n} \left[\op{\Set{S,\phi}}{\Set{S',\phi}}\otimes \ketbra{\phi}^{\otimes t_2}\right]\\
            = \E_{\ket{\phi}\gets \mu_n} \left[\E_{\theta\gets [0,2\pi]}\left[\op{\Set{S,e^{i\theta}\phi}}{\Set{S',e^{i\theta}\phi}}\otimes e^{t_2 i\theta} e^{-t_2 i \theta} \ketbra{\phi}^{\otimes t_2}\right]\right]\\
            = \E_{\ket{\phi}\gets \mu_n} \left[\E_{\theta\gets [0,2\pi]}\left[\op{\Set{S,e^{i\theta}\phi}}{\Set{S',e^{i\theta}\phi}}\otimes \ketbra{\phi}^{\otimes t_2}\right]\right]\\
        \end{split}
    \end{equation}
    But we also have 
    \begin{equation}\label{eq:help2}
        \begin{split}
            \op{\Set{S,e^{i\theta}\phi}}{\Set{S',e^{i\theta}\phi}} = \bigotimes_{i=1}^{t_1}(\Ind_{i \notin S \cap S'} \ketbra{0} + e^{i\theta}\Ind_{i \in S\setminus S'} \op{\phi}{0} \\
            + e^{-i\theta}\Ind_{i \in S'\setminus S}\op{0}{\phi} + \Ind_{i \in S'\cap S}\ketbra{\phi})\\
            = (e^{i\theta})^{|S\setminus S'|} (e^{-i\theta})^{|S'\setminus S|}\op{\Set{S,\phi}}{\Set{S',\phi}}\\
            = e^{(|S| - |S'|)i \theta}\op{\Set{S,\phi}}{\Set{S',\phi}}
        \end{split}
    \end{equation}
    and so by~\Cref{eq:help1,eq:help2}, and since $|S| - |S'| \neq 0$,
    \begin{equation}\label{eq:help3}
        \begin{split}
            \E_{\ket{\phi}\gets \mu_n} \left[\op{\Set{S,\phi}}{\Set{S',\phi}}\otimes \ketbra{\phi}^{\otimes t_2}\right]\\
            = \E_{\ket{\phi}\gets \mu_n} \left[\E_{\theta\gets [0,2\pi]}\left[e^{(|S| - |S'|)) i\theta}\op{\Set{S,\phi}}{\Set{S',\phi}}\otimes  \ketbra{\phi}^{\otimes t_2}\right]\right]\\
            = \E_{\ket{\phi}\gets \mu_n}[0]\\
            =0
        \end{split}
    \end{equation}
    And so by~\Cref{eq:help0,eq:help3}, we have $\rho = \rho'$.
\end{proof}

We construct $\SSM$ from $\SSMM$ for the specific case when $\D$ is the Haar distribution.

\begin{construction}\label{cons:ssmmtossm}
    Define $C^{\SSMM}_n$ as follows:
    \begin{enumerate}
        \item Run $\SSMM_n\to \ket{\psi}$. 
        \item Apply the measurement $\{\ketbra{0},I-\ketbra{0}\}$ to $\ket{\phi}$.
        \item If the result is $I-\ketbra{0}$, output the residual state.
        \item Otherwise, repeat.
        \item If this process fails $m$ times, output $\bot$.
    \end{enumerate}
\end{construction}

\begin{theorem}\label{thm:ssmmtossmindiff}
    For all $T_{\Sim},T_1,T_2,m$ satisfying $T_1\leq T_{\Sim}$, for $\D$ the Haar distribution over $n$ qubits, $C^{\SSMM}$ is $\left(T_{\Sim},T_1,T_2,\frac{T_2}{2^m}\right)$ indifferentiable from $\SSM$.
\end{theorem}

\begin{proof}
    The main idea behind the simulator is that it directly constructs the state $\rho'$ from~\Cref{lem:binom}. The simulator is defined as follows. 
    \begin{construction}\label{cons:statesim}
    The simulator $\Sim_n^{\SSM}$ acts as follows on initialization
    \begin{enumerate}
        \item Sample $c \gets B(T_{\Sim}(n),1/2)$ (samples $c$ from the binomial distribution).
        \item Run $\SSM_n$ $c$ times to produce $c$ copies of $\ket{\phi_n}$ in registers $B_1,\dots,B_{c}$.
        \item Construct the state
        $$\ket{\Rep{T_{Sim},c,\ket{1}}}_{A_1,\dots,A_{T_{\Sim}}}\ket{\phi_n}^{\otimes c}_{B_1,\dots,B_c}\ket{1}_D$$
        \item For each qubit $i$ of $A$, controlled on $A_i$ being $0$, do a SWAP between $A_i$ and $B_D$ and increment register $D$ by one. This produces the state
        $${T_{\Sim} \choose c}^{-1/2}\left(\sum_{S\subseteq [T_{\Sim}]:|S| = c} \bigotimes_{i = 1}^{T_{\Sim}} (\Ind_{i \notin S}\ket{0}_{A_i} + \Ind_{i \in S}\ket{\phi}_{A_i})\right) \otimes \ket{1}^{\otimes c}_{B_1,\dots,B_c}\ket{c+1}_D$$
        Note that this is exactly $\ket{\Rep{T_{\Sim},c,\phi_n}}\ket{1}^{\otimes c}\ket{c+1}$.
    \end{enumerate}
    On query $i$, $\Sim_n$ outputs registers $A_i$.
    \end{construction}
    \vspace{0.2in}

    To show indifferentiability, we will work through a series of hybrid games.
    Let $\A^{(\cdot,\cdot)}$ be any adversary making at most $T_1,T_2$ queries to its oracles respectively.
    \begin{enumerate}
        \item Hybrid 1: $\A^{\Sim^{\SSM},\SSM}$, the ideal world in the indifferentiability game.
        \item Hybrid 2: $\A^{\widetilde{\Sim}, \SSM}$, the same as hybrid 1, but we replace the state $\ket{\Rep{T_{\Sim},c,\phi_n}}\ket{c+1}$ used by the simulator with $\ket{\phi_n-}^{\otimes k}\ket{c+1}$. In particular, $\widetilde{\Sim}$ behaves exactly the same as $\Sim$, but its internal state after sampling $c$ is set to
        $$\bigotimes_{i=1}^k \frac{1}{\sqrt{2}}\left(\ket{0}_{A_i}\ket{0^n}_{B_i} - \ket{1}_{A_i}\ket{\phi_n}_{B_i}\right)\otimes \ket{0^n}^{\otimes c}_{C_1,\dots,C_c}\ket{c+1}_D$$
        \item Hybrid 3: $\A^{\SSMM,\SSM}$, the same as hybrid 1, but the simulator is replaced with the honest $\SSMM$ oracle. In an abuse of notation, in this hybrid both $\SSM$ and $\SSMM$ are going to refer to the same state $\ket{\phi}\gets \D$.
        \item Hybrid 4: $\A^{\SSMM, C^{\SSMM}}$, the real world in the indifferentiability game.
    \end{enumerate}

    We then proceed to show that each of these hybrids are close.
    \begin{enumerate}
        \item $\Pr[\A^{\Sim^{\SSM},\SSM} \to 1] = \Pr[\A^{\widetilde{\Sim}^{\SSM}, \SSM} \to 1]$. This follows directly from~\Cref{lem:binom}. In particular, the mixed state representing the internal state of the simulator along with all responses from $\SSM$ has the exact same density matrix in both games.
        \item $\Pr[\A^{\widetilde{\Sim}^{\SSM}, \SSM} \to 1] = \Pr[\A^{\SSMM, \SSM} \to 1]$. This is because during query $i \leq T_{\Sim}$ made to $\wt{\Sim}$, registers $A_i$ contain $\ket{\phi_n-}$, which is exactly what is returned by $\SSMM$. Furthermore, since $T_1 \leq T_{\Sim}$, we will always have $i\leq T_1 \leq T_{\Sim}$ and so register $A_i$ will always be assigned.
        \item $\abs{\Pr[\A^{\SSMM, \SSM} \to 1] - \Pr[\A^{\SSMM, C^{\SSMM}} \to 1]} \leq \frac{T_2}{2^m}$. Note that as long as $C^{\SSM}$ does not fail, it will output $\ket{\phi_n}$. Thus, the only way to distinguish these two sets of oracles is if $C^{\SSMM}$ ever fails. But by union bound, this probability is $\leq \frac{T_2}{2^m}$.
    \end{enumerate}

    And so we get that $C^{\SSMM}$ is $\left(T_{\Sim},T_1,T_2,\frac{T}{2^m}\right)$ indifferentiable from $\SSM$.
\end{proof}

\begin{corollary}\label{cor:ssmmtossm}
    Let $m=2n$ and let $C^{\SSMM}$ be the construction from~\Cref{cons:ssmmtossm}. For all primitives $\Prim^{\SSM}$ and cryptographic games $G$, if $\Prim^{\SSM}$ is $(c,\epsilon)$-secure under $G$ relative to $\SSM$, then $\Prim^{C^{\SSMM}}$ is $(c,\epsilon+\frac{1}{2^n})$-secure under $G$ relative to $\SSM$.
\end{corollary}

\begin{proof}
    This follows from~\Cref{cor:compcor}. In particular, observe that for all polynomials $p$, for all sufficiently large $n$,
    $$\frac{p(n)}{2^{2n}} \leq \frac{1}{2^n}$$
\end{proof}
\section{$\SSMM$ and $\Swap$ are equivalent}

We now construct $\Swap$ from $\SSMM$.

\begin{proposition}[Adapted from Theorem 4 of\cite{C:JiLiuSon18}]\label{prop:reflsim}
    Let $t \in \N$. Let $\ket{\psi}$ be a quantum state of dimension $d$ and let $R_{\ket{\psi}} = I - 2\ketbra{\psi}$ be reflection around $\ket{\psi}$. Let $\Pi^{Sym}$ be the projection onto the symmetric subspace $\vee^{t+1} \C^d$ and let $R_{\Pi^{Sym}} = I - 2\Pi^{Sym}$ be reflection around the symmetric subspace. 
    
    Define the quantum channel $Q$ to, on input $\rho$, apply $R_{\Pi^{Sym}}$ to $\rho \otimes \ket{\psi}^{\otimes t}$ and output the first register. Then
    $$\norm{Q - R_{\ket{\psi}}\cdot R_{\ket{\psi}}}_{\diamond} \leq \frac{2}{\sqrt{t+1}}$$
\end{proposition}

We define $C^{\SSMM}$ to simply run the algorithm from the lemma on $n$ states sampled from $\SSMM$.

\begin{theorem}\label{thm:ssmmtoswap}
    For all $T_1,T_2$, $t$. $C^{\SSMM}$ is $\left(1,T_1,T_2,\frac{2T_2}{\sqrt{t+1}}\right)$ reset indifferentiable from $\Swap$.
\end{theorem}

\begin{proof}
    The simulator $\Sim_n^{\Swap}$ here will be defined as follows:
    \begin{enumerate}
        \item Denote $\ket{\phi_n}$ to be the state which $\Swap$ swaps with $0$.
        \item Begin with the state $\ket{-}\ket{0}$.
        \item Apply $\Swap$ controlled on the first qubit to the second qubit.\\
        This produces the state $\frac{1}{\sqrt{2}}\ket{0}\ket{0}-\frac{1}{\sqrt{2}}\ket{1}\ket{\phi_n}$.
        \item Apply $X$ to the first qubit controlled on the second qubit being $0$.\\
        This produces the state $\ket{1}\left(\frac{1}{\sqrt{2}}\ket{0}\ket{0}-\frac{1}{\sqrt{2}}\ket{1}\ket{\phi_n}\right)$.
        \item Output the second register. Up to a global phase this state is exactly $\ket{\phi_n-}$.
    \end{enumerate}
    In particular, this is an exact, stateless simulator. By making a single oracle query to $\Swap$, it produces the state $\ket{\phi_n-}$.

    Applying induction to ~\Cref{prop:reflsim} shows that for all adversaries $\A$ making at most $T_2$ queries to the second oracle,
    $$\abs{\Pr[\A^{\SSMM, C^{\SSMM}} \to 1] - \Pr[\A^{\SSMM, \Swap} \to 1]} \leq \frac{2T_2}{\sqrt{t+1}}$$

    But $\Sim^{\Swap}$ and $\SSMM$ are identical, and so we have
    $$\abs{\Pr[\A^{\SSMM, C^{\SSMM}} \to 1] - \Pr[\A^{\Sim^{\Swap}, \Swap} \to 1]} \leq \frac{2T_2}{\sqrt{t+1}}$$
\end{proof}

\begin{corollary}\label{cor:simswapsmmm}
    Since the simulator is stateless, the converse also holds. That is, for all $t,T_1,T_2$, $\Sim^{\Swap}$ is $\left(t,T_1,T_2,\frac{2T_1}{\sqrt{t+1}}\right)$ reset indifferentiable from $\SSMM$.
\end{corollary}
\section{Barriers to simulating $\ket{\phi-}$ using $\ket{\phi}$.}

We first observe that given any polynomial number of copies of $\ket{\phi}$, it is impossible to, in a black-box manner, produce $\ket{\phi-}$. This follows essentially immediately from our~\Cref{lem:binom}. In particular,~\Cref{lem:binom} shows that the mixed state
$$\rho = \ketbra{\phi}^{\otimes t_1}\otimes \E_{c\gets B(t_2,1/2)}\left[\ketbra{\Rep{t_2,c,\phi}}\right]$$
is indistinguishable from the pure state
$$\rho' = \ketbra{\phi}^{\otimes t_1}\otimes \ketbra{\phi-}^{\otimes t_2}$$

But note that if you could produce $\ket{\phi-}$ from $\ket{\phi}^{\otimes t_1}$, then you could distinguish these two states by doing a swap test between the produced state $\ket{\phi-}$ and the last register.

And so it is impossible to directly produce $\ket{\phi-}$ from $\ket{\phi}$. This does not rule out the ability to produce $\ket{\psi-}$ for some distinct but identically distributed state $\ket{\psi}$.

We also observe that if we consider $\D$ the distribution over states to be something other than the Haar distribution, there are some tasks which are possible with oracle access to $\ket{\phi-}$ but not with $\ket{\phi}$. In particular, let $\D$ be the distribution $e^{i\theta} \ket{1}$ for a random phase $\theta$. Given $t$ copies of $\ket{\phi}^{\otimes t} = (e^{i\theta} \ket{1})^{\otimes t}$, it is impossible to compute $\theta$ since $\theta$ appears only in the global phase $e^{ti\theta}$. On the other hand, given $t$ copies of $$\ket{\phi-}^{\otimes t} = \left(\frac{1}{\sqrt{2}}\ket{0} - \frac{e^{i\theta}}{\sqrt{2}}\ket{1}\right),$$
it is possible to approximate $\theta$ by guessing the value of $\theta$ and performing swap tests.

Thus, if we consider the game where two parties must agree on a shared random bit with no communication, it is possible for two parties with oracle access to $\ket{\phi-}$ to win this game, but not for two parties with oracle access to $\ket{\phi}$.

Overall, while we do not definitively rule out any indifferentiable construction of $\SSMM$ from $\SSM$ for Haar random states, we do rule out most obvious approaches. Any construction must explicitly use the fact that the states are Haar random, and must transform the given state in some non-trivial way. 
\section{Simulating SWAP in the LOCC model}\label{sec:locc}

In this section, we show how to simulate the $\SSMM$ model using the $\SSM$ model against a restricted class of adversaries for the specific case where $\D$ is the Haar distribution. Our main idea behind building an "indifferentiable" version of $\SSMM$ model using the $\SSM$ model is to reverse the roles of simulator and construction in~\Cref{thm:ssmmtossmindiff}. Note that this leaves us with a \textit{stateful} construction, which is typically not allowed. However, by restricting the indifferentiability adversary to be LOCC, we show that a composition theorem holds for this construction whenever the cryptographic game is itself LOCC.

In particular, let $C^{\SSM}$ be the simulator from~\Cref{cons:statesim}. As a reminder, each query to $C^{\SSM}$ will output a fresh register of
$$\E_{\substack{c\gets B(t,1/2)}}\Big[\ketbra{\Rep{t,c,\phi}}\Big]$$

We proceed to prove the following main theorem
\begin{theorem}\label{cor:finalloccindiff}
    For all $\ell_\A,T_1,T_2=\poly(n)$, $C^{\SSM}$ simulating $t$ copies is \\
    $(\ell_\A, n, T_1, T_2, O(\ell_\A (T_1+T_2)^5/\sqrt{2^n}))$-LOCC indifferentiable from $\SSMM$ as long as $T_2 \leq t$.
\end{theorem}

The key idea behind this proof will rely on a key lemma (\Cref{lem:keylemma}), which informally says that any pair of LOCC adversaries cannot tell if they are given two copies of this construction or one shared copy. The rest of the proof proceeds by a hybrid argument, given in full in~\Cref{sec:loccfinal}.

To show the key lemma, we rely on the technique of~\cite{TCC:AnaGulLin24}, which observes that two states are LOCC indistinguishable if their partial transposes are close in trace distance. We then explicitly compute the difference in the partial transposes of the relevant density matrices. Like in~\cite{TCC:AnaGulLin24}, bounding this value can be reduced to bounding the spectrum of a graph with a particular structure. We then bound this value by relying on techniques for analyzing Kneser graphs.
\subsection{Representing Haar random states}

We give a number of definitions regarding type vectors, a notion closely connected to representations of Haar random states.

A $t$-copy type over $[N]$ is a multiset $T$ over elements in $[N]$. A vector $v \in [N]^t$ has type $T$ if for all $x \in [N]$, $x$ appears $k$ times in $v$ if and only if $x$ has multiplicity $k$ in $T$. We write $type(v) = T$. We define a type vector
    $$\ket{T} \propto \sum_{type(v) = T}\ket{v}$$

    In the spirit of~\cite{TCC:AnaGulLin24}, we say a type $T$ is collision free if the multiplicity of all of its elements is at most $1$. Then
    $$\ket{T} = \frac{1}{\sqrt{t!}} \sum_{type(v) = T}\ket{v}$$

    We denote $Ty(t,X)$ to be the set of $t$-copy types over $X$. We denote ${X \choose t}$ to be the set of $t$-copy collision-free types over $X$, i.e. the subsets of $X$ of size $t$.

    We will often consider types over $[N] \cup \{0\}$ with special significance added to the $0$ locations. We will often require that all $0$ locations appear in some subset of the registers.

    Let $T \subseteq [N]$ be any type. We define $T^{0(t)}$ to be the type defined by adding $0$s to $T$ until it reaches $t$ elements. Formally,
    $$T^{0(t)} \coloneq T \cup \underbrace{\{0,\dots,0\}}_{t-|T|}$$

    Let $A,B$ be any registers containing $|A|,|B|$ states of dimension $N+1$ respectively. Let $T$ be any type satisfying $|A| \leq |T| \leq |A| + |B|$. We will define
    $$\ket{\Zero(T,B)}_{AB} \propto \sum_{\substack{|v| = |A|+|B|\\type(v) = T^{0(|A|+|B|)}\\v_i \neq 0\text{ for }i\in A}} \ket{v}$$
    which will represent the type vector $T^{0(|A|+|B|)}$ restricted to all the $0$s appearing in the subregisters of $B$.

    For any collision-free subset $T\subseteq [N]$,
    $$\ket{\Zero(T,B)}_{AB} = \frac{1}{\sqrt{|T|!{t \choose |T|-|A|}}}\sum_{\substack{|v|=|A|+|B|\\type(v) = T^{0(|A|+|B|)}\\v_i \neq 0\text{ for }i\in A}} \ket{v}$$

    We also extend the requirement that all $0$s appear in some subset of the registers to the case where a certain amount of $0$s appear in one subset and a different amount in another. In particular, we define
    $$\ket{\Zero^2(T,B_1,B_2)^{b^f_1,b^f_2}}_{AB_1B_2} \propto \sum_{\substack{|v| = |A|+|B_1|+|B_2|\\type(v) = T^{0(|A|+|B_1|+|B_2|)}\\v_i \neq 0\text{ for }i\in A\\\#\text{non-}0s\text{ in }v_{B_1}=b_1^f\\\#\text{non-}0s\text{ in }v_{B_2}=b_2^f}} \ket{v}$$
    Note that for a collision-free type $T$,
    $$\ket{\Zero^2(T,B_1,B_2)^{b^f_1,b^f_2}}_{AB_1B_2} = \frac{1}{\sqrt{|T|!{b_1 \choose b_1^f}{b_2 \choose b_2^f}}}\sum_{\substack{|v| = |A|+|B_1|+|B_2|\\type(v) = T^{0(|A|+|B_1|+|B_2|)}\\v_i \neq 0\text{ for }i\in A\\\#\text{non-}0s\text{ in }v_{B_1}=b_1^f\\\#\text{non-}0s\text{ in }v_{B_2}=b_2^f}} \ket{v}$$

\begin{lemma}\label{lem:zerosplit}
    Let $A_1,A_2,B_1,B_2$ be registers containing $a_1,a_2,b_1,b_2$ states of dimension $N+1$ respectively. For all collision-free $T$,
    \begin{equation*}
        \begin{split}
            \ket{\Zero(T,(B_1,B_2))}_{A_1,A_2,B_1,B_2} \\
            = \sum_{\substack{X \subseteq T\\|X|\in [a_1,a_1+b_1]\\|T|-|X|\in [a_2,a_2+b_2]}} \sqrt{\alpha_{|X|}} \ket{\Zero(X,B_1)}_{A_1,B_1}\ket{\Zero(T\setminus X,B_2)}_{A_2,B_2}
        \end{split}
    \end{equation*}
    where
    $$\alpha_i = \frac{{b_1 \choose i-a_1}{b_2 \choose |T|-i-a_2}}{{b_1+b_2 \choose |T|-a_1-a_2}{|T| \choose i}}$$
\end{lemma}

\begin{lemma}\label{lem:zerosplit2}
    Let $A_1,A_2,B_1,B_2$ be registers containing $a_1,a_2,b_1,b_2$ states of dimension $N+1$ respectively. For all collision-free $T$,
    \begin{equation}
        \begin{split}
            \ket{\Zero^2(T,B_1,B_2)^{b_1^f,b_2^f}}_{A_1,A_2,B_1,B_2}\\
            = \frac{1}{\sqrt{{|T|\choose a_1 + b_1^f}}}\sum_{\substack{X \subseteq T\\|X|=a_1+b_1^f}} \ket{\Zero(X,B_1)}_{A_1 B_1}\ket{\Zero(T\setminus X,B_2)}_{A_2 B_2}
        \end{split}
    \end{equation}
\end{lemma}

The proofs of these two lemmas are deferred to~\Cref{sec:lemmas}.

\begin{theorem}\label{lem:typevector}
    Let $t,n\in \N$. Then,
    $$\E_{\ket{\phi} \gets \mu_{n}}\left[\ketbra{\phi}^{\otimes t}\right] = \E_{T \gets Ty(t, [2^n])}\left[\ketbra{T}\right]$$
\end{theorem}
\subsection{Key Lemma}

Let $\cal{D}$ be the Haar distribution. We will define $C^{\SSM}$ to be exactly the simulator $\Sim^{\SSM}$ from~\Cref{cons:statesim}.

\begin{lemma}[Key Lemma]\label{lem:keylemma}
    Let $A_1,B_1,A_2,B_2$ be registers containing $a_1,b_1,a_2,b_2$ respectively states of dimension $N+1$, satisfying $N\geq (a_1+a_2+b_1+b_2+1)^2$.
    Define $\wt{\rho}_{(A_1,B_1),(A_2,B_2)}$ to be the state
    $$\E_{\substack{c\gets B(b_1+b_2,1/2)\\T\gets {[N] \choose a_1+a_2 + c}}}\left[\ketbra{\Zero(T,(B_1,B_2))}_{A_1,A_2,B_1,B_2}\right]$$
    and $\wt{\sigma}_{(A_1,B_1),(A_2,B_2)}$ to be the state
    $$\E_{\substack{c_1\gets B(b_1,1/2)\\c_2\gets B(b_2,1/2)\\T\gets {[N] \choose a_1+a_2 + c_1 + c_2}}}\left[\ketbra{\Zero^2(T,B_1,B_2)^{c_1,c_2}}_{A_1,A_2,B_1,B_2}\right]$$
    Then $\norm{\wt{\rho}^{\Gamma^{A_2B_2}}-\wt{\sigma}^{\Gamma^{A_2B_2}}}_1 \leq \frac{e(a_1+a_2+b_1+b_2)^6}{\sqrt{N}}$.
\end{lemma}

\begin{proof}
    \begin{equation*}
        \begin{split}
            \wt{\rho}=\sum_{c\in [b_1+b_2]}\sum_{T \in {[N] \choose a_1+a_2+c}} \frac{{b_1+b_2\choose c}}{2^{b_1+b_2}} \frac{1}{{N\choose a_1+a_2+c}} \ketbra{\Zero(T,(B_1,B_2))}_{A_1A_2B_1B_2}
        \end{split}
    \end{equation*}
    Setting
    \begin{align}
            \alpha_{c} &= \frac{{b_1+b_2\choose c}}{2^{b_1+b_2}} \frac{1}{{N\choose a_1+a_2+c}}&
            \beta_{c,c_1} &= \frac{{b_1 \choose c_1}{b_2 \choose c-c_1}}{{b_1+b_2 \choose c}{a_1+a_2+c \choose a_1+c_1}}&
            \gamma_{c,c_1,c_2} &= \alpha_c\sqrt{\beta_{c,c_1}\beta_{c,c_2}}
    \end{align}
    and applying~\Cref{lem:zerosplit} gives
    \begin{equation*}
        \sum_{\substack{c\in [b_1+b_2]\\c_1,c_2\leq c}}\gamma_{c,c_1,c_2}\sum_{\substack{T \in {[N] \choose a_1+a_2+c}\\X,Y\subseteq T\\|X|=|Y|=a_1+c_1}}\ket{\Zero(X,B_1)}\bra{\Zero(Y,B_1)}_{A_1B_1}\otimes \ket{\Zero(T\setminus X,B_2)}\bra{\Zero(T\setminus Y,B_2)}_{A_2B_2}
    \end{equation*}
    Observe that when $c_1 > b_1$ or $c_2 > b_2$, $\gamma_{c,c_1,c_2} = 0$. And so
    \begin{equation}
        \begin{split}
            &\wt{\rho}^{\Gamma^{(A_2,B_2)}}\\
            &=
        \sum_{\substack{c\in [b_1+b_2]\\c_1,c_2\leq c}}\gamma_{c,c_1,c_2}\sum_{\substack{T \in {[N] \choose a_1+a_2+c}\\X,Y\subseteq T\\|X|=a_1+c_1\\|Y|=a_2+c_2}}\ket{\Zero(X,B_1)}\bra{\Zero(Y,B_1)}_{A_1B_1}\otimes \ket{\Zero(T\setminus Y,B_2)}\bra{\Zero(T\setminus X,B_2)}_{A_2B_2}
        \end{split}
    \end{equation}

    We can similarly expand out $\wt{\sigma}$.
    \begin{equation}
        \begin{split}
            \wt{\sigma}_{(A_1,B_1),(A_2,B_2)}\\
            =\sum_{\substack{c_1\in[b_1]\\c_2\in [b_2]}}\frac{{b_1\choose c_1}{b_2\choose c_2}}{2^{b_1+b_2}{N\choose a_1+a_2+c_1+c_2}}\sum_{T\in {[N]\choose a_1+a_2+c_1+c_2}}\ket{\Zero^2(T,B_1,B_2)^{c_1,c_2})}\bra{\Zero^2(T,B_1,B_2)^{c_1,c_2}}\\
            =\sum_{\substack{c \in [b_1+b_2]\\c_1\leq c}}\frac{{b_1\choose c_1}{b_2\choose c-c_1}}{2^{b_1+b_2}{N\choose a_1+a_2 + c}}\sum_{T\in {[N]\choose a_1+a_2+c}}\ket{\Zero^2(T,B_1,B_2)^{c_1,c_2})}\bra{\Zero^2(T,B_1,B_2)^{c_1,c_2}}
        \end{split}
    \end{equation}
    Applying~\Cref{lem:zerosplit2} gives
    \begin{equation}
        \begin{split}
            \wt{\sigma}_{(A_1,B_1),(A_2,B_2)}\\
            =\sum_{\substack{c \in [b_1+b_2]\\c_1\leq c}}\frac{{b_1\choose c_1}{b_2\choose c-c_1}}{2^{b_1+b_2}{N\choose a_1+a_2 + c}}\frac{1}{{a_1+a_2+c \choose a_1+c_1}}\\
            \sum_{\substack{T\in {[N]\choose a_1+a_2+c}\\X,Y\subseteq T\\|X|=|Y|=a_1+c_1}}\ket{\Zero(X,B_1)}\bra{\Zero(Y,B_1)}_{A_1B_1}\otimes \ket{\Zero(T\setminus X,B_2)}\bra{\Zero(T\setminus Y,B_2)}_{A_2B_2}
        \end{split}
    \end{equation}

    As $$\frac{{b_1\choose c_1}{b_2\choose c-c_1}}{2^{b_1+b_2}{N\choose a_1+a_2 + c}}\frac{1}{{a_1+a_2+c \choose a_1+c_1}} = \gamma_{c,c_1,c_1},$$
    subtraction then gives
    \begin{equation}\label{eq:diff}
        \begin{split}
            \wt{\rho}^{\Gamma^{A_2B_2}} - \wt{\sigma}^{\Gamma^{A_2B_2}}\\
            =\sum_{\substack{c\in [b_1+b_2]\\c_1\neq c_2\leq c}}\gamma_{c,c_1,c_2}\sum_{\substack{T \in {[N] \choose a_1+a_2+c}\\X,Y\subseteq T\\|X|=a_1+c_1\\|Y|=a_1+c_2}}\begin{array}{c}
                \ket{\Zero(X,B_1)}\bra{\Zero(Y,B_1)}_{A_1B_1}\\
                \otimes \ket{\Zero(T\setminus Y,B_2)}\bra{\Zero(T\setminus X,B_2)}_{A_2B_2}
            \end{array}
        \end{split}
    \end{equation}

    We then observe (again in the style of~\cite{TCC:AnaGulLin24}) that for each $X,Y$, if we set $I=X\cap Y$, $C = T\setminus (X\cup Y)$, $X'=X\setminus I$, $Y'=Y\setminus I$, then
    \begin{equation}\label{eq:counting}
        \begin{split}
            \ket{\Zero(X,B_1)}\bra{\Zero(Y,B_1)}_{A_1B_1}\otimes \ket{\Zero(T\setminus Y,B_2)}\bra{\Zero(T\setminus X,B_2)}_{A_2B_2} \\
            = \ket{\Zero(I\cup X',B_1)}\bra{\Zero(I\cup Y',B_1)}_{A_1B_1}\otimes \ket{\Zero(C\cup X',B_2)}\bra{\Zero(C\cup Y',B_2)}_{A_2B_2}
        \end{split}
    \end{equation}

    We will define
    \begin{equation}\label{eq:taudef}
        \tau_{c_1,c_2,I,C} = \sum_{\substack{X\in {[N]\setminus (I\cup C) \choose a_1+c_1-|I|}\\Y\in {[N]\setminus (I\cup C) \choose a_1+c_2-|I|}\\X\cap Y=\emptyset}}\begin{array}{c}
        \ket{\Zero(I\cup X',B_1)}\bra{\Zero(I\cup Y',B_1)}_{A_1B_1}\\
        \otimes \ket{\Zero(C\cup X',B_2)}\bra{\Zero(C\cup Y',B_2)}_{A_2B_2}
        \end{array}
    \end{equation}
    ~\Cref{eq:diff,eq:counting,eq:taudef} together give us
    \begin{equation}
        \begin{split}
            \wt{\rho}^{\Gamma^{A_2B_2}} - \wt{\sigma}^{\Gamma^{A_2B_2}}\\
            =\sum_{\substack{c\in [b_1+b_2]\\c_1\neq c_2\leq c}}\gamma_{c,c_1,c_2}\sum_{i=0}^{a_1+\min(c_1,c_2)}\sum_{j=0}^{a_2+c-\max(c_1,c_2)}\sum_{I\in {[N]\choose i}}\sum_{C\in {[N]\setminus I\choose j}}\tau_{c_1,c_2,I,C}
        \end{split}
    \end{equation}

    We can then directly bound the trace norm of $\tau_{c_1,c_2,I,C}$ as follows
    \begin{corollary}\label{cor:tauboundmain}
        $$\norm{\tau_{c_1,c_2,I,C}}_1 \leq (a_1+c)N^{a_1+\frac{c_1+c_2}{2}-|I|}$$
    \end{corollary}

    We defer the proof of this to~\Cref{app:kneser}. The rough idea is similar to~\cite{TCC:AnaGulLin24}, in that we relate the spectrum of $\tau_{c_1,c_2,I,C}$ to that of a graph, and then directly compute the spectrum of the graph. The graph in question is a modified variant of the Kneser graph used by~\cite{TCC:AnaGulLin24}, and the argument used to bound its spectrum follows a similar approach to works bounding the spectrum of the Kneser graph~\cite{kneserwayback}.

    Thus, we get that
    \begin{equation}\label{eq:bigmath}
        \begin{split}
            \norm{\wt{\rho}^{\Gamma^{A_2B_2}} - \wt{\sigma}^{\Gamma^{A_2B_2}}}_1\\
            \leq \sum_{\substack{c\in [b_1+b_2]\\c_1\neq c_2\leq c}}\gamma_{c,c_1,c_2}\sum_{i=0}^{a_1+\min(c_1,c_2)}\sum_{j=0}^{a_2+c-\max(c_1,c_2)}\sum_{I\in {[N]\choose i}}\sum_{C\in {[N]\setminus I\choose j}}\norm{\tau_{c_1,c_2,I,C}}_1\\
            \leq \sum_{\substack{c\in [b_1+b_2]\\c_1\neq c_2\leq c}}\gamma_{c,c_1,c_2}\sum_{i=0}^{a_1+\min(c_1,c_2)}\sum_{j=0}^{a_2+c-\max(c_1,c_2)}{N\choose i}{N\choose j}(a_1+c)N^{a_1-i+\frac{c_1+c_2}{2}}\\
            \leq \sum_{\substack{c\in [b_1+b_2]\\c_1\neq c_2\leq c}}\gamma_{c,c_1,c_2}\sum_{i=0}^{a_1+\min(c_1,c_2)}\sum_{j=0}^{a_1+c-\max(c_1,c_2)}(a_1+c) N^i\cdot N^j\cdot N^{a_1-i+\frac{c_1+c_2}{2}}\\
            \leq \sum_{\substack{c\in [2t]\\c_1\neq c_2\leq c}}(a_1+c)\gamma_{c,c_1,c_2}\sum_{i=0}^{t+\min(c_1,c_2)}\sum_{j=0}^{t+c-\max(c_1,c_2)} N^{a_1+j+(c_1+c_2)/2}\\
        \end{split}
    \end{equation}
    But note that since $c_1\neq c_2$, we have
    $\max(c_1,c_2)\geq \min(c_1,c_2)+1$
    and so
    $\frac{c_1+c_2}{2}=\frac{\max(c_1,c_2)+\min(c_1,c_2)}{2}\leq \max(c_1,c_2)-\frac{1}{2}$
    and so expanding on~\Cref{eq:bigmath} we get
    \begin{equation}\label{eq:bigmath2}
        \begin{split}
            \norm{\wt{\rho}^{\Gamma^{A_2B_2}} - \wt{\sigma}^{\Gamma^{A_2B_2}}}_1\\
            \leq \sum_{\substack{c\in [b_1+b_2]\\c_1\neq c_2\leq c}}\gamma_{c,c_1,c_2}(a_1+c)\sum_{i=0}^{a_1+\min(c_1,c_2)}\sum_{j=0}^{a_2+c-\max(c_1,c_2)} N^{a_1+j+\max(c_1,c_2)-1/2}\\
            \leq \sum_{\substack{c\in [b_1+b_2]\\c_1\neq c_2\leq c}}\gamma_{c,c_1,c_2}(a_1+c)(a_1+\min(c_1,c_2))(a_2+c-\max(c_1,c_2)) N^{a_1+a_2+c-1/2}\\
            \leq \sum_{\substack{c\in [b_1+b_2]\\c_1\neq c_2\leq c}}((a_1+c)^2(a_2+c)\cdot N^{a_1+a_2+c-1/2}\alpha_c)\sqrt{\beta_{c,c_1}\beta_{c,c_2}}\\
        \end{split}
    \end{equation}

    We then compute out that
    \begin{equation}\label{eq:neglbound}
        \begin{split}
            (a_1+c)^2(a_2+c)N^{a_1+a_2+c-1/2}\alpha_c\\
            = (a_1+c)^2(a_2+c)\frac{{b_1+b_2\choose c}}{2^{b_1+b_2}}\frac{N^{a_1+a_2+c-1/2}}{{N\choose a_1+a_2+c}}\\
            \leq (a_1+c)^2(a_2+c)\frac{{b_1+b_2\choose c}}{2^{b_1+b_2}}\frac{N^{a_1+a_2+c-1/2}}{(N-a_1-a_2-c)^{a_1+a_2+c}}\\
            =(a_1+c)^2(a_2+c)\frac{{b_1+b_2\choose c}}{2^{b_1+b_2}}\left(\frac{N}{N-a_1-a_2-c}\right)^{a_1+a_2+c}\frac{N^{a_1+a_2+c-1/2}}{N^{a_1+a_2+c}}\\
            \leq (a_1+c)^2(a_2+c)\frac{{b_1+b_2\choose c}}{2^{b_1+b_2}}\left(\frac{N}{N-a_1-a_2-c}\right)^{a_1+a_2+c}\frac{1}{\sqrt{N}}
        \end{split}
    \end{equation}
    And since $N>>a_1,a_2$, we have $N\geq (a_1+a_2+c+1)^2$ and so
    \begin{equation}\label{eq:neglbound2}
        \begin{split}
            \frac{N}{a_1+a_2+c} \geq a_1+a_2+c+1\\
            N \leq N-a_1-a_2-c+\left(\frac{N-a_1-a_2-c}{a_1+a_2+c}\right)\\
            \frac{N}{N-a_1-a_2-c} \leq 1 + \frac{1}{a_1+a_2+c}\\
            \left(\frac{N}{N-a_1-a_2-c}\right)^{a_1+a_2+c} \leq e
        \end{split}
    \end{equation}
    Combining~\Cref{eq:neglbound,eq:neglbound2}, we get
    \begin{equation}\label{eq:neglboundf}
        (a_1+c)^2(a_2+c) N^{a_1+a_2+c-1/2}\alpha_c \leq \frac{(a_1+c)(a_2+c)e}{\sqrt{N}}
    \end{equation}

    Combining~\Cref{eq:bigmath2,eq:neglboundf}, we get
    \begin{equation}
        \begin{split}
            \norm{\wt{\rho}^{\Gamma^{A_2B_2}} - \wt{\sigma}^{\Gamma^{A_2B_2}}}_1
            \leq \sum_{\substack{c\in [b_1+b_2]\\c_1\neq c_2\leq c}}\frac{(a_1+c)^2(a_2+c)e}{\sqrt{N}}\sqrt{\beta_{c,c_1}\beta_{c,c_2}}\\
            \leq \frac{(a_1+b_1+b_2)^2(a_2+b_1+b_2)e}{\sqrt{N}}\sum_{\substack{c\in [b_1+b_2]\\c_1\neq c_2\leq c}}\sqrt{\beta_{c,c_1}\beta_{c,c_2}}\\
            \leq \frac{(a_1+b_1+b_2)^2(a_2+b_1+b_2)e}{\sqrt{N}}\sum_{\substack{c\in [b_1+b_2]\\c_1\neq c_2\leq c}}1\\
            \leq (a_1+b_1+b_2)^2(a_2+b_1+b_2)(b_1+b_2)^3\frac{e}{\sqrt{N}}\leq \frac{e(a_1+a_2+b_1+b_2)^6}{\sqrt{N}}
        \end{split}\qedhere
    \end{equation}
\end{proof}

\subsection{Proof of LOCC indifferentiability given key lemma}\label{sec:loccfinal}

\begin{lemma}[\cite{TCC:AnaGulLin24} Proof of Theorem 7.9]\label{lem:loccbound}
    Let $\rho_{AB},\sigma_{AB}$ be two states over registers $A$ and $B$. The LOCC distinguishing advantage between $\rho_{AB}$ and $\sigma_{AB}$ is bounded by
    $$\frac{1}{2}\norm{\rho^{\Gamma^{B}} - \sigma^{\Gamma^{B}}}_1$$
\end{lemma}

\begin{theorem}\label{thm:mainthm}
    Let $n\in \N$ and let $N = 2^n$. Let $A_1,B_1,A_2,B_2$ be register containing $a_1,b_1,a_2,b_2$ respectively states of dimension $N+1$ satisfying $N \geq (a_1+a_2+b_1+b_2+1)^2$.
    Define $\rho_{(A_1,B_1),(A_2,B_2)}$ to be the state
    $$\E_{\substack{\ket{\phi} \gets \mu_n\\c\gets B(b_1+b_2,1/2)}}\left[\ketbra{\phi}^{\otimes a_1}_{A_1} \otimes \ketbra{\phi}^{\otimes a_2}_{A_2} \otimes \ketbra{\Rep{b_1+b_2,c,\phi}}_{B_1,B_2}\right]$$
    and $\sigma_{(A_1,B_1),(A_2,B_2)}$ to be the state
    $$\E_{\substack{\ket{\phi} \gets \mu_n\\c_1\gets B(b_1,1/2)\\c_2 \gets B(b_2,1/2)}}\left[\ketbra{\phi}^{\otimes a_1}_{A_1} \otimes \ketbra{\phi}^{\otimes a_2}_{A_2} \otimes \ketbra{\Rep{b_1,c_1,\phi}}_{B_1} \otimes \ketbra{\Rep{b_2,c_2, \phi}}_{B_2}\right]$$
    Then $\rho,\sigma$ are $O\left(\frac{(a_1+a_2+b_1+b_2)^5}{\sqrt{N}}\right)$-LOCC indistinguishable.
\end{theorem}

\begin{proof}
    By~\Cref{lem:loccbound}, the LOCC distinguishing advantage between $\rho_{A_1B_1A_2B_2}$ and $\sigma_{A_1B_1A_2B_2}$ is upper bounded by
    $$\frac{1}{2}\norm{\rho^{\Gamma^{A_2B_2}} - \sigma^{\Gamma^{A_2B_2}}}_1$$

    Note that we can construct $\rho$ and $\sigma$ by applying an isometry to $$\E_{\substack{\ket{\phi} \gets \mu_n\\c\gets B(b_1+b_2,1/2)}}\left[\ketbra{\phi}^{\otimes a_1+a_2+c}\right]$$ and $$\E_{\substack{\ket{\phi} \gets \mu_n\\c_1\gets B(b_1,1/2)\\c_2 \gets B(b_2,1/2)}}\left[\ketbra{\phi}^{\otimes a_1+a_2+c_1+c_2}\right]$$ respectively. Applying~\Cref{lem:typevector}, we can generate $\rho$ and $\sigma$ by applying the same isometry to 
    $$\E_{\substack{c\gets B(b_1+b_2,1/2)\\T\gets Ty([N], a_1+a_2+c)}}\left[\ketbra{T}\right]$$
    and
    $$\E_{\substack{c_1\gets B(b_1,1/2)\\c_2 \gets B(b_2,1/2)\\T\gets Ty([N], a_1+a_2+c_1+c_2)}}\left[\ketbra{T}\right]$$ 

    This leads to the observation that
    $$\rho = \E_{\substack{c\gets B(2t,1/2)\\T\gets Ty([N],a_1+a_2+c)}}\left[\ketbra{\Zero(T,(B_1,B_2))}_{A_1,A_2,B_1,B_2}\right]$$
    and
    $$\sigma = \E_{\substack{c_1\gets B(t,1/2)\\c_2\gets B(t,1/2)\\T\gets Ty([N],a_1+a_2+c_1+c_2)}}\left[\ketbra{\Zero^2(T,B_1,B_2)^{c_1,c_2}}_{A_1,A_2,B_1,B_2}\right]$$

    Let $\wt{\rho},\wt{\sigma}$ be from~\Cref{lem:keylemma}. Since the probability that a given type $T \gets Ty([N],t)$ is not collision-free is $\leq \frac{t^2}{N}$ by the birthday bound, we have
    \begin{equation}
        \begin{split}
            \norm{\rho - \wt{\rho}}_1 = O\left(\frac{(a_1+a_2+b_1+b_2)^2}{2^n}\right)\\
            \norm{\sigma - \wt{\sigma}}_1 = O\left(\frac{(a_1+a_2+b_1+b_2)^2}{2^n}\right)
        \end{split}
    \end{equation}

    By the triangle inequality and~\Cref{lem:keylemma},
    $$\norm{\rho^{\Gamma^{A_2B_2}} - \sigma^{\Gamma^{A_2B_2}}}_1 \leq O\left(\frac{(a_1+a_2+b_1+b_2)^6}{\sqrt{N}}\right)$$
    which is negligible in $n$.

    And so therefore $\rho,\sigma$ are LOCC indistinguishable.
\end{proof}

\begin{lemma}\label{lem:hybtechlemma}
    Define 
    $$\Pi_c^t = \sum_{\substack{S\subseteq [t]\\|S|=c}}\bigotimes_{i=1}^t\left(\Ind_{i\notin S}(\ketbra{0}) + \Ind_{i\in S}(I - \ketbra{0})\right)$$
    and let
    $M^t = \{\Pi^t_i\}_{i \in [t]}$. That is, $M^t$ is the measurement which counts the number of non-zero registers out of $t$ registers.

    Let $\ket{\phi}$ be any state and let $A_1,\dots,A_\ell,B_1,\dots,B_\ell$ be registers each containing $a_1,\dots,a_\ell,b_1,\dots,b_\ell$ states respectively. Let $b = b_1+\dots+b_\ell,a=a_1+\dots+a_\ell$ Let $\rho$ be the state
    $$\E_{c \gets B(b,1/2)}\left[\ketbra{\phi}^{\otimes a}_{A_1\dots A_\ell} \otimes \ketbra{\Rep{b,c,\phi}}_{B_1\dots B_\ell}\right]$$
    Let $\sigma_{A_1A_2}$ be the state produced by measuring the number of copies of $\ket{\phi}$ contained in register $B_1,\dots,B_i$. That is, the mixed state resulting from applying the measurement $M^{b_i}$ to each register $B_1,\dots,B_i$ of $\rho$. Then,
    \begin{equation*}
        \begin{split}
            \sigma = \E_{\substack{c_j \gets B(b_j,1/2),j\leq i\\ c \gets B(b_{i+1}+\dots+b_\ell, 1/2)}}\bigg[&\ketbra{\phi}^{\otimes a}_{A_1\dots A_\ell}\\
            &\otimes \ketbra{\Rep{b_1,c_1,\phi}}_{A_1} \otimes \dots \otimes \ketbra{\Rep{b_i,c_i,\phi}}_{A_i}\\
            &\otimes \ketbra{\Rep{b_{i+1}+\dots+b_\ell,c,\phi}}\bigg]
        \end{split}
    \end{equation*}
\end{lemma}

\begin{proof}
    For ease of presentation, we will omit the $A$ registers and consider the case $i=1$, $\ell=2$. The full proof is analogous.

    In this case, we will work with
    $$\rho = \E_{c \gets B(b,1/2)}\left[\ketbra{\Rep{b,c,\phi}}_{B_1B_2}\right]$$
    and $\sigma$ the result of applying $M^t$ to register $B_1$.

    Let us first fix $c$, and compute the mixed state resulting from applying $M^{b_1}$ to $\ketbra{\Rep{b,c,\phi}}$. Note that
    \begin{equation}
        \begin{split}
            \ket{\Rep{b,c,\phi}} = {b\choose c}^{-1/2}\sum_{S\in {[b]\choose c}} \ket{\Set{b,S,\phi}}\\
            ={b\choose c}^{-1/2} \sum_{c_1 = 0}^{c}\sum_{S_1 \in {[b_1]\choose c_1}}\sum_{S_2 \in {[b_2] \choose c-c_1}}\ket{\Set{b_1,S_1,\phi}}\ket{\Set{b_2,S_2,\phi}}\\
            ={b\choose c}^{-1/2} \sum_{c_1 = 0}^{c}\left(\sum_{S_1 \in {[b_1] \choose c_1}}\ket{\Set{b_1,S_1,\phi}}\right)\otimes \left(\sum_{S_2 \in {[b_2] \choose c-c_1}}\ket{\Set{b_2,S_2,\phi}}\right)\\
            ={b\choose c}^{-1/2} \sum_{c_1=0}^{c}\left({b_1\choose c_1}^{1/2}\ket{\Rep{b_1,c_1,\phi}}\right)\otimes \left({b_2\choose c-c_1}^{1/2}\ket{\Rep{b_2,c_2,\phi}}\right)\\
            = \sum_{c_1=0}^c \sqrt{\frac{{b_1\choose c_1}{b_2\choose c-c_1}}{{b\choose c}}}\ket{\Rep{b_1,c_1,\phi}}\ket{\Rep{b_2,c-c_1,\phi}}
        \end{split}
    \end{equation}

    But observe that for all $b_1,c_1$, $$\ket{\Rep{a_1,c_1,\phi}}\in \Pi^{b_1}_{c_1}.$$ And so applying the measurement $M^{a_1}$ on the $B_1$ register of $\ketbra{\Rep{b_1+b_2,c,\phi}}_{B_1B_2}$ produces the mixed state
    \begin{equation}
        \begin{split}
            \sum_{c_1=0}^c \frac{{b_1\choose c_1}{b_2\choose c_2}}{{b_1+b_2\choose c}} \ketbra{\Rep{b_1,c_1,\phi}}\otimes \ketbra{\Rep{b_2,c-c_1,\phi}}
        \end{split}
    \end{equation}

    By linearity, we have that $\sigma$ is equal to
    \begin{equation}
        \begin{split}
            \sum_{c = 0}^{b_1+b_2} \frac{{b_1+b_2 \choose c}}{2^{b_1+b_2}} \left(\sum_{c_1=0}^c \frac{{b_1\choose c_1}{b_2\choose c_2}}{{b_1+b_2\choose c}} \ketbra{\Rep{b_1,c_1,\phi}}\otimes \ketbra{\Rep{b_2,c-c_1,\phi}}\right)\\
            = \sum_{c_1=0}^{b_1}\sum_{c_2=0}^{c_2}\frac{{b_1\choose c}{b_2\choose c}}{2^{b_1+b_2}}\ketbra{\Rep{b_1,c_1,\phi}}\ketbra{\Rep{b_2,c_2,\phi}}\\
            =\E_{\substack{c_1 \gets B(b_1,1/2)\\ c_2 \gets B(b_2, 1/2)}}\left[\ketbra{\Rep{b_1,c_1,\phi}}_{B_1} \otimes \ketbra{\Rep{b_2,c_2,\phi}}_{B_2}\right]
        \end{split}
    \end{equation}
\end{proof}

\begin{corollary}\label{cor:hybindis}
    Let $\ell,n\in \N$ with $N = 2^n$. Let $A_1,\dots, A_{\ell},B_1,\dots,B_{\ell}$ be registers containing $a_1,\dots,a_\ell,b_1,\dots,b_\ell$ respectively states of dimension $N+1$. Let $a = a_1+\dots+a_\ell,b=b_1+\dots+b_\ell$. We will require that $N>(a+b+1)^2$.
    
    Define $\rho_{(A_1,B_1),\dots,(A_\ell,B_\ell)}^{a,b}$ to be the state
    $$\E_{\substack{\ket{\phi} \gets \mu_n\\c\gets B(b,1/2)}}\left[\ketbra{\phi}^{\otimes a}_{A_1\dots A_\ell} \otimes \ketbra{\Rep{\ell t, c,\phi}}_{B_1,\dots,B_\ell}\right]$$
    and $\sigma_{(A_1,B_1),\dots,(A_\ell,B_\ell)}^{a,b}$ to be the state
    $$\E_{\substack{\ket{\phi} \gets \mu\\c_1,\dots,c_\ell \gets B(t,1/2)}}\left[\ketbra{\phi}^{\otimes a}_{A_1\dots A_\ell} \otimes \ketbra{\Rep{b_1,c_1,\phi}}_{B_1} \otimes\dots  \otimes \ketbra{\Rep{b_\ell,c_\ell, \phi}}_{B_\ell}\right]$$
    Then $\rho^{\ell},\sigma^{\ell}$ are $O(\ell(a+b)^5)/\sqrt{2^n})$-LOCC indistinguishable.
\end{corollary}

\begin{proof}
This follows a hybrid argument. In particular, let $Hyb_i$ be
\begin{equation}
    \begin{split}
        \E_{\substack{\ket{\phi} \gets \mu_n\\c_j\gets B(b_j,1/2),j\leq i\\c \gets B(b_{i+1}+\dots+b_\ell, 1/2)}}\bigg[&\ketbra{\phi}^{\otimes a}_{A_1\dots A_\ell}\\  &\otimes \ketbra{\Rep{b_1,c_1,\phi}}_{B_1} \otimes\dots  \otimes \ketbra{\Rep{b_i,c_i, \phi}}_{B_i}\\
        &\otimes \ketbra{\Rep{b_{i+1}+\dots+b_\ell,c,\phi}}_{B_{i+1}\dots B_\ell}\bigg]
    \end{split}
\end{equation}
We will show that for all $i$, $Hyb_i$ and $Hyb_{i+1}$ are LOCC indistinguishable.

Let $\A$ be any algorithm distinguishing $Hyb_i$ and $Hyb_{i+1}$ with advantage $\epsilon$. We will construct an algorithm $\A'$ distinguishing $\rho$ and $\sigma$ from~\Cref{thm:mainthm} with advantage $\epsilon$.

The $A_1$ and $B_1$ registers from~\Cref{thm:mainthm} will consist of the registers $A_1,\dots,A_{i+1}$ and $B_1,\dots,B_{i+1}$ respectively. The $A_2$ and $B_2$ registers will consist of the registers $A_{i+2},\dots,A_\ell$ and $B_{i+2},\dots,B_\ell$ respectively.

$\A'$ will operate as follows: for each $1\leq j\leq i$, apply the measurement $M^{b_i}$ from ~\Cref{lem:hybtechlemma} to register $B_i$. Then, run $\A$, passing party $\A_j$ registers $A_jB_j$, and produce the same output. 

By~\Cref{lem:hybtechlemma}, if the initial state was $\rho$, then measuring the number of copies of $\ket{\phi}$ in registers $B_1,\dots,B_i$ produces exactly $Hyb_i$. 

Similarly, if the initial state was $\sigma$, then measuring the number of copies of $\ket{\phi}$ in the registers $B_1,\dots,B_i$ produces exactly $Hyb_{i+1}$.

Thus, $\A'$'s distinguishing advantage is exactly the same as $\A$'s, and so by~\Cref{thm:mainthm}, $Hyb_i$ and $Hyb_{i+1}$ are $O((a+b)^5/\sqrt{2^n})$-LOCC indistinguishable.

By the triangle inequality, $\rho^{a,b}$ and $\sigma^{a,b}$ are $O(\ell(a+b)^5)/\sqrt{2^n})$-LOCC indistinguishable.
\end{proof}

\begin{theorem}[\Cref{cor:finalloccindiff} restated]
    For all $\ell_\A,T_1,T_2=\poly(n)$, $C^{\SSM}$ simulating $t$ copies is \\
    $(\ell_\A, n, T_1, T_2, O(\ell_\A (T_1+T_2)^5/\sqrt{2^n}))$-LOCC indifferentiable from $\SSMM$ as long as $T_2 \leq t$.
\end{theorem}

\begin{proof}
    We first define the simulator $\Sim^{\SSMM}$ to be the construction from~\Cref{cons:ssmmtossm} instantiated with $m=n$. That is, the simulator samples $\SSMM$ and measures if the state is $0$. If it is, it samples again, otherwise, it returns the state.

    We will prove this via a sequence of hybrids.
    \begin{enumerate}
        \item $Hyb_1$: The real game, where $\A^{(\cdot)} = (\A_1^{(\cdot)},\dots,\A_\ell^{(\cdot)})$ has query access to the oracles\\
        $((\SSM,C_1^{\SSM}),\dots,(\SSM,C_\ell^{\SSM}))$ where each $C_i^{\SSM}$ is a different instance of the construction.
        \item $Hyb_2$: Where we allow the construction to share state between each LOCC party. That is, we replace the oracle with \\
        $((\SSM,C^{\SSM}),\dots,(\SSM,C^{\SSM}))$.
        \item $Hyb_3$: The ideal game, where $\A^{\cdot}$ has query access to the oracles \\$((\Sim^{\SSMM},\SSMM),\dots,(\Sim^{\SSMM},\SSMM))$.
    \end{enumerate}

    The argument for each hybrid goes as follows:

    \begin{enumerate}
        \item $\abs{\Pr[Hyb_1\to 1] - \Pr[Hyb_2 \to 1]} \leq O\left(\frac{\ell\cdot T^5}{\sqrt{2^n}}\right)$:\\
        this follows immediately from~\Cref{cor:hybindis}. In particular, the oracle $((\SSM,C_1^{\SSM}),\dots,(\SSM,C_\ell^{\SSM}))$ can be simulated by giving each $\A_i$ query access to $(A_i,B_i)$ of $\rho^{a,b}$.
        \item $\abs{\Pr[Hyb_2 \to 1] - \Pr[Hyb_3 \to 1]} \leq \frac{T_2}{2^n}$:\\
        this follows from exactly the same argument as~\Cref{thm:ssmmtossmindiff}, but with the role of construction and simulator swapped.
    \end{enumerate}

    By the triangle inequality,
    \begin{equation}
        \begin{split}
            \bigg|\Pr[\A^{((\SSM,C_1^{\SSM}), \dots, (\SSM,C_t^{\SSM}))} \to 1] - \\\Pr[\A^{((\Sim_1^{\SSMM},\SSMM), \dots, (\Sim_5^{\SSMM},\SSMM))} \to 1]\bigg| \leq O\left(\frac{\ell\cdot T^5}{\sqrt{2^n}}\right)
        \end{split}
    \end{equation}
\end{proof}
\section{Applications}
We formalize a number of useful theorems which can be used to directly apply our results in practice.

\begin{corollary}[Constructions in $\SSMM$ to $\Swap$]\label{corf:ssmmtoswap}
    Let $\D = \{\D_n\}_{n\in\N}$ be any balanced family of distributions over $n$ qubit states parameterizing $\SSMM,\Swap$. Let $G$ be any efficient security game. Let $d\in\N$ and let $\epsilon:\N\to [0,1]$. Let $\Sim^{\Swap}$ be the construction from~\Cref{cor:simswapsmmm}.

    Let $\Prim^{\SSMM}$ be a primitive. Define 
    $$\ol{\Prim}^{\Swap} \coloneqq \Prim^{\Sim^{\Swap}}.$$
    
    If $\Prim^{\SSMM}$ is $(\epsilon;c)$-secure under $G$ relative to $\SSMM$, then $\ol{\Prim}^{\Swap}$ is $(\epsilon + n^{-d};c)$-secure under $G$ relative to $\Swap$.
\end{corollary}

\begin{proof}
    This follows from~\Cref{cor:simswapsmmm,cor:compcor} by setting $t=q=(pn^d)^2$.
\end{proof}

\begin{corollary}[Constructions in $\Swap$ to $\SSMM$]\label{corf:swaptossmm}
    Let $\D = \{\D_n\}_{n\in\N}$ be any balanced family of distributions over $n$ qubit states parameterizing $\SSMM,\Swap$. Let $G$ be any efficient security game making at most $T_{G,1},T_{G,2}\leq \poly(n)$ queries to the primitive and adversary respectively. Let $d\in\N$ and let $\epsilon:\N\to [0,1]$. 

    Let $\Prim^{\Swap}$ be a primitive making at most $T_{\Prim}$ queries to its oracle.
    
    Let $C^{\SSMM}$ be the construction from~\Cref{thm:ssmmtoswap} with $t = (2T_{G,1}\cdot T_{\Prim}\cdot n^d)^2$. Define 
    $$\ol{\Prim}^{\SSMM} \coloneqq \Prim^{C^{\SSMM}}.$$
    
    If $\Prim^{\Swap}$ is $(\epsilon;c)$-secure under $G$ relative to $\Swap$, then $\ol{\Prim}^{\SSMM}$ is $(\epsilon + n^{-d};c)$-secure under $G$ relative to $\SSMM$.
\end{corollary}

\begin{proof}
    This follows from~\Cref{thm:ssmmtoswap,thm:indiffcomp}. In particular, we want to show that for all $T_{\A} \leq \poly(n)$, $\ol{\Prim}^{\Swap}$ is $(t,\epsilon+n^{-d};c)$ secure under $G$.

    We know that $\Prim^{\Swap}$ is $(T_{\A},\epsilon;c)$ secure under $G$. Furthermore,~\Cref{thm:ssmmtoswap} gives that $C^{\SSMM}$ is $\left(1, T_{G,2}\cdot T_{\A},T_{G,1}\cdot T_{\Prim},\frac{2T_{G,1}\cdot T_{\Prim}}{\sqrt{t+1}} \leq n^{-d}\right)$-indifferentiable from $\mathcal{Q}$.

    And so, by~\Cref{thm:indiffcomp}, we get that $\ol{\Prim}^{\SSMM}$ is $(T_{\A},\epsilon + n^{-d};c)$-secure under $G$ relative to $\SSMM$.
\end{proof}

\begin{corollary}[Constructions in $\SSM$ to $\SSMM$]\label{corf:ssmtossmm}
    Let $\D = \{\D_n\}_{n\in\N}$ be any balanced family of distributions over $n$ qubit states parameterizing $\SSMM,\SSM$. Let $G$ be any efficient single-stage security game. Let $\epsilon:\N\to [0,1]$. Let $C^{\SSMM}$ be the construction from~\Cref{cons:ssmmtossm} with $m=2n$.

    Let $\Prim^{\SSM}$ be a primitive. Define 
    $$\ol{\Prim}^{\SSMM} \coloneqq \Prim^{\Sim^{\SSMM}}.$$
    
    If $\Prim^{\SSM}$ is $(\epsilon;c)$-secure under $G$ relative to $\SSMM$, then $\ol{\Prim}^{\SSMM}$ is $\left(\epsilon + \frac{1}{2^n};c\right)$-secure under $G$ relative to $\SSMM$.
\end{corollary}

\begin{proof}
    This follows immediately from~\Cref{cor:ssmmtossm,cor:compcor}.
\end{proof}

\begin{corollary}[Constructions in $\SSM$ to $\Swap$]\label{corf:ssmtoswap}
    Let $G$ be any single-stage security game. If there exists a primitive $\Prim^{\SSM}$ which is $c$-secure under $G$ relative to $\SSM$, then there exists a primitive $\ol{\Prim}^{\Swap}$ which is $c$-secure under $G$ relative to $\Swap$.
\end{corollary}

\begin{proof}
    This follows from~\Cref{corf:ssmtossmm,corf:ssmmtoswap} and the observation that since $d$ can be arbitrary, the constructions actually preserve negligible security.
\end{proof}

\begin{corollary}[Constructions in $\SSMM$ to $\SSM$]\label{corf:ssmmtossm}
    Let $\SSM,\SSMM$ be parameterized by $\D = \{\mu_n\}_{n\in\N}$, the Haar distribution over $n$ qubits. Let $d\in\N$ and let $\epsilon:\N\to [0,1]$. Let $G$ be any efficient LOCC security game making at most $T_{G,1},T_{G,2}\leq \poly(n)$ queries to the primitive and adversary respectively. 

    Let $\Prim^{\Swap}$ be a primitive making at most $T_{\Prim}$ queries to its oracle. 
    
    Let $C^{\SSMM}$ be the construction from~\Cref{cor:finalloccindiff} with $t \geq T_{G,1}\cdot T_{\Prim}$. Define 
    $$\ol{\Prim}^{\SSMM} \coloneqq \Prim^{C^{\SSMM}}.$$
    
    If $\Prim^{\Swap}$ is $(\epsilon;c)$-secure under $G$ relative to $\Swap$, then $\ol{\Prim}^{\SSMM}$ is $(\epsilon + O(2^{-n/2});c)$-secure under $G$ relative to $\SSMM$.
\end{corollary}

\begin{proof}
    This follows immediately from~\Cref{cor:finalloccindiff,thm:locccomp}.
\end{proof}

\paragraph*{Sample applications} Note that existing works show that for $\D$ a Haar random state
\begin{enumerate}
    \item One-way state generators~\cite{EPRINT:BosCheNeh24,EPRINT:BMMMY24}, multi-copy pseudorandom states~\cite{ARXIV:CheColSat24,TCC:AnaGulLin24}, and key exchange and commitments with classical communication~\cite{TCC:AnaGulLin24} do not exist relative to $\SSM$.
    \item One-way puzzles~\cite{EPRINT:BosCheNeh24}, one-copy pseudorandom states~\cite{ARXIV:CheColSat24,TCC:AnaGulLin24}, and quantum commitments exist relative to $\SSM$.
\end{enumerate}

These composition theorems then give us (with a little bit of work) that all those statements are true when $\SSM$ is replaced by $\Swap$.
\begin{enumerate}
    \item One-way state generators, multi-copy pseudorandom states, and key exchange and commitments with classical communication do not exist relative to $\Swap$.
    \item One-way puzzles, one-copy pseudorandom states, and quantum commitments exist relative to $\Swap$.
\end{enumerate}

In detail, ~\cite{EPRINT:BosCheNeh24}'s argument that one-way state generators do not exist in the $\SSM$ can easily be adapted to the $\SSMM$ and to handle inverse polynomial error. Our results then immediately give that one-way state generators do not exist relative to $\Swap$. Full details of this argument are given in~\Cref{sec:owsg,sec:noowsgssmm}.

Note that this immediately implies that multi-copy pseudorandom states do not exist in the $\Swap$ model, since pseudorandom states can be used to build one-way state generators~\cite{C:MorYam22,ARXIV:CGG+24}.

Furthermore, our~\Cref{corf:ssmtoswap} immediately gives that key exchange and commitments with classical communication do not exist relative to $\Swap$, since their security games are LOCC and neither primitive exists relative to $\SSM$~\cite{TCC:AnaGulLin24}. To show this explicitly, we formally define the security game for classical communication key exchange in~\Cref{sec:ke}.

On the other hand, all positive constructions (one-way puzzles, one-copy pseudorandom states, and quantum commitments) translate to the $\Swap$ model immediately via~\Cref{corf:ssmtoswap}. Note that there is a subtlety here, which is that one-way puzzle security cannot be described by an efficient cryptographic game. However, one-way puzzles are equivalent (via black box reduction) to state puzzles, a primitive introduced by~\cite{ARXIV:KhuTom24} with an efficient security game. And so, through this reduction, it can also be observed that one-way puzzles exist in the $\Swap$ model.
\subsection{OWSG do not exist in the $\SSM $ model for any distribution}\label{sec:noowsgssmm}

A one-way state generator is a procedure $(StateGen, Ver)$, where $StateGen$ takes in a classical key $k$ and produces an output state $\rho_k$. $Ver$ takes in an output state $\rho$ and a key $k$ and checks if $k$ matches $\rho$. Security says that given any polynomially many copies of $\rho_k$, it should be hard to find a $k'$ which passes verification. A game-based definition is given in~\Cref{sec:owsg}.

It is known from~\cite{EPRINT:BosCheNeh24} that one-way state generators do not exist in the $\SSM$ model. In order to adapt this impossibility to the $\Swap$ model, it turns out that it is sufficient to adapt the attack from~\cite{EPRINT:BosCheNeh24} to the $\SSMM$ model with bounded polynomial copy security and inverse polynomial error. In particular, we prove a slightly more general result
\begin{lemma}[Proof adapted from Theorem 4.1 of~\cite{EPRINT:BosCheNeh24}]\label{lem:noowsg}
    For any family of distributions $\D$, there does not exist a $\left(O(n^2), \frac{1}{10n^2}\right)$-one-way state generator relative to $\SSM$.
\end{lemma}

with the following corollary.
\begin{corollary}[Proof in~\Cref{sec:owsg}]
    There does not exist a one-way state generator relative to $\Swap$.
\end{corollary}

The attack from~\cite{EPRINT:BosCheNeh24} utilizes a threshold search procedure to find the key for the one-way state generator from many copies of the state. Formally, the threshold search procedure is defined as follows.

\begin{definition}[Threshold Search Problem~\cite{STOC:BadODo21}]
    Let $\{M_i\}_{i \in [m]}$ be a collection of $2$-outcome measurements. Let $\rho$ be any quantum state with the promise that there exists some $k$ such that
    $$\Tr(M_i\rho) \geq 3/4$$
    The threshold search problem is to, given some number of copies of $\rho$, output $k$ such that $\Tr(M_i\rho)\geq 1/3$.
\end{definition}

\begin{lemma}[Randomized Threshold Search Procedure~\cite{ARXIV:WatBos22}]
    There is an algorithm which uses $O(\log^2(m))$ space and copies of $\rho$ which solves the threshold search problem with constant probability.
\end{lemma}

We adapt their attack to an arbitrary $\SSM$ setting (and inverse polynomial security error) as follows.

Let $OWSG^{\SSM} = (StateGen^{\SSM},Ver^{\SSM})$. Let $T_{Ver}$ be an upper bond on the number of queries made by $Ver$.

Let $\ket{\phi_n}$ be the state returned by $\SSM$.

Since in the $\SSM$ model, the oracle has no input, we can assume that $Ver(k,\rho)$ applies some unitary $U_k$ to the state $\rho_A \otimes \left(\ket{\phi_n}^{\otimes T_{Ver}}\right)_B$ and then measures a bit in the computational basis.

We define $\Pi_k$ to be the following operator acting on the input for $U_k$ on registers $AB$,
$$\Pi_k = U^{\dagger}_k(\ketbra{1} \otimes I)U_k$$
That is, $\Pi_k$ is the operator which checks if $Ver(k,\rho)$ accepts.

\begin{lemma}\label{lem:thresholda}
    Consider the following experiment
    \begin{enumerate}
        \item Sample $k\gets \{0,1\}^n$
        \item Sample $StateGen^{\SSM}(k) \to \rho_k$, $10n$ times
        \item Check if $\Pi_k^{\otimes 10n}$ accepts $\left(\rho_k \otimes \ket{\phi_n}^{\otimes T_{Ver}}\right)^{\otimes 10n}$
    \end{enumerate}
    Then the probability that this process accepts is at least $1-\frac{1}{n}$.
\end{lemma}

\begin{proof}
    Note that this experiment is just the correctness game for the one-way state generator played $10n$ times with the same $k$. In particular, if $BAD_i$ is the event that the $i$th copy of the verifier rejects, then correctness says that $\Pr[BAD_i] \leq \frac{1}{10n^2}$. And so by the union bound, $\Pi_k$ accepts with probability $\geq 1 - \frac{10n}{10n^2} = 1-\frac{1}{10n}$.
\end{proof}

\begin{lemma}\label{lem:thresholdb}
    If for any key $k$ and fixed state $\rho$, if the probability that $\Pi_k^{\otimes 10n}$ accepts $\left(\rho \otimes \ket{\phi_n}^{\otimes T_{Ver}}\right)^{\otimes 10n}$ is $\geq 1/3$, then the probability that $\Pi_k$ accepts $\rho \otimes \ket{\phi_n}^{\otimes T_{Ver}}$ is $\geq 1 - \frac{1}{n}$.
\end{lemma}

\begin{proof}
    Let $p = \Tr\left(\Pi_k \left(\rho \otimes \ket{\phi_n}^{\otimes T_{Ver}}\right)\right)$ be the probability that $Ver(k,\cdot)$ accepts $\rho$. Then $p^{10n} \geq \frac{1}{3}$ and so we get
    \begin{equation}
        \begin{split}
            p \geq \frac{1}{3^{1/10n}} = e^{-\frac{\ln 3}{10n}} \geq 1 - \frac{\ln 3}{10n} \geq 1 - \frac{1}{n}
        \end{split}
    \end{equation}
\end{proof}

These two lemmas together show that for an honestly sampled $StateGen(k)\to \rho_k$, the threshold search procedure will find an accepting preimage. 

In particular, the algorithm to find $k$ given $\rho_k^{O(n^2)}$ is to run threshold search for $\{\Pi_k\}_{k\in 2^n}$ on $O(n)$ copies of $\left(\rho_k\otimes \ket{\phi_n}^{\otimes T_{Ver}}\right)^{\otimes 10n}$. Note that in this case, $m=2^n$, and so threshold search requires $O(\log(2^n))=O(n)$ copies.

From~\Cref{lem:thresholda}, we know that for $StateGen(k)\to \rho_k$, 
$$\Tr\left(\Pi_k \left(\rho_k\otimes \ket{\phi_n}^{\otimes T_{Ver}}\right)^{\otimes 10n}\right) \geq \frac{3}{4}$$ 
and so the promise of threshold search is met. And so, with constant probability, this algorithm finds some $k'$ such that 
$$\Tr\left(\Pi_{k'} \left(\rho_k\otimes \ket{\phi_n}^{\otimes T_{Ver}}\right)^{\otimes 10n}\right) \geq \frac{1}{3}$$
By~\Cref{lem:thresholdb}, this means that $Ver(k',\rho_k)$ accepts with probability at least $1-\frac{1}{n}$. 

In total, this algorithm wins the game $G_{OWSG}^{sec-O(n^2)}$ with probability $c\cdot (1-1/n)$ for some constant $c$ (the constant of threshold search). In particular, this is greater than $\frac{1}{10n^2}$, and so~\Cref{lem:noowsg} follows.

\section{Acknowledgments}
We would like to thank Aditya Gulati for pointing out a linear algebra mistake in an earlier version of this paper.

\bibliographystyle{alpha} 
\bibliography{abbrev0,crypto,bib}

\appendix
\section{Indifferentiability Composition Theorems}\label{sec:comp}

\begin{theorem}[Bounded query composition theorem][\Cref{thm:indiffcomp} restated]
Let $\mathcal{P},\mathcal{Q}$ be two idealized primitives, and let $C^{\mathcal{P}}$ be a construction $(T_{\Sim},T_1,T_2,\delta)$-indifferentiable from $\mathcal{Q}$. Let $G$ be any single-stage cryptographic game making at most $T_{G,1}$ queries to its primitive and $T_{G,2}$ queries to its adversary, and let $\Prim^{\mathcal{Q}}$ be any primitive relative to $\mathcal{Q}$ making at most $T_{\Prim}$ queries to its oracle. Let $\epsilon,T_{\A}:\N\to [0,1]$ be any functions and let $c$ be any constant.

As long as $T_{G,1}\cdot T_{\Prim} \leq T_2$ and $T_{G,2}\cdot T_{\A} \leq T_1$, then if $\Prim^{\cal{Q}}$ is $(T_{\A}\cdot T_{\Sim},\epsilon;c)$-secure under $G$ relative to $\cal{Q}$, then $\Prim^{C^{\cal{P}}}$ is $\left(T_{\A}, \epsilon + \delta;c\right)$-secure under $G$ relative to $\cal{P}$.
\end{theorem}

\begin{proof}
    Given any adversary $\mathcal{A}^{\mathcal{P}}$ making at most $T_{\A}$ queries to its oracle, we can define an adversary $\mathcal{B}^{\mathcal{Q}}$ which simulates $\mathcal{A}$, and whenever $\mathcal{A}$ makes a query to $\mathcal{P}$, $\mathcal{B}$ answers it by running $\Sim^{\mathcal{Q}}$. Note that $\mathcal{B}$ makes at most $T_{\A} \cdot T_{\Sim}$ queries to its oracle.

    Note that $G(1^n,\Prim^{C^{\mathcal{P}}},\mathcal{A}^{\mathcal{P}})$ is an algorithm making at most $T_{G,1}\cdot T_{\Prim} \leq T_2$ queries to the construction and at most $T_{G,2}\cdot T_{\A} \leq T_1$ direct queries to $\mathcal{P}$.

    By $(T_{\Sim},T_1,T_2,\epsilon)$-indifferentiability,
    $$\abs{\Pr[G(1^n, \Prim^{C^{\mathcal{P}}},\mathcal{A}^{\mathcal{P}},)] - \Pr[G(1^n, \Prim^{\mathcal{Q}},\mathcal{A}^{\Sim^{\mathcal{Q}}})]} \leq \delta(n)$$

    But note that $\mathcal{A}^{\Sim^{\mathcal{Q}}} = \mathcal{B}^{\mathcal{Q}}$. By $(T_{\A}\cdot T_{\Sim},c,\epsilon)$-security of $\Prim$ under $G$ relative to $\cal{Q}$,
    $$\Pr[G(1^n,\Prim^{\mathcal{Q}},\mathcal{B}^{\mathcal{Q}})] - c \leq \epsilon(n)$$
    
    By the transitive property, we have 
    $$\Pr[G(1^n,\Prim^{C^{\mathcal{P}}},\mathcal{A}^{\mathcal{P}})] - c\leq \epsilon(n) + \delta(n)$$
\end{proof}

\begin{corollary}[General composition theorem][\Cref{cor:compcor} restated]
    Let $\mathcal{P},\mathcal{Q}$ be two idealized primitives, let $C^{\mathcal{P}}$ be a construction of $\mathcal{Q}$ from $\mathcal{P}$, and let $\delta:\N\to [0,1]$. If, for all $p = \poly(n)$, there exists some $q = \poly(n)$ such that $C^{\mathcal{P}}$ is $(q,p,p,\delta)$-indifferentiable from $\mathcal{Q}$, then the following holds:

    For all primitives $\Prim^{\mathcal{Q}}$ and single-stage cryptographic games $G$, if $\Prim^{\mathcal{Q}}$ is $(\epsilon;c)$-secure under $G$ relative to $\mathcal{Q}$, then $\Prim^{C^{\mathcal{P}}}$ is $(\epsilon+\delta;c)$-secure under $G$ relative to $\mathcal{P}$.
\end{corollary}

\begin{proof}
    Let $\A^{\mathcal{P}}$ be any adversary making at most $T_\A$ queries to its oracle. Let $T_{G,1},T_{G,2}\leq \poly(n)$ be an upper bound for the number of queries to $\Prim,\A$ respectively by $G(\Prim,\A)$. Note that the values of $T_{G,1},T_{G,2}$ may depend on $\A$.
    
    Let $p=\max(T_{G,1}\cdot T_{\Prim}, T_{G,2}\cdot T_\A)$. We know that for some $q\leq \poly(n)$, $C^{\mathcal{P}}$ is $(q,p,p,\delta)$-indifferentiable from $\mathcal{Q}$. Furthermore, since $T_{\A}\cdot q\leq \poly(n)$, we know that $\Prim^{\mathcal{Q}}$ is $(T_\A\cdot q,\epsilon; c)$ secure under $G$ relative to $\mathcal{Q}$.
    
    And so the exact same argument as~\Cref{thm:indiffcomp} gives that 
    $$\Pr[G(1^n,\Prim^{C^{\mathcal{P}}},\mathcal{A}^{\mathcal{P}})] - c\leq \epsilon(n) + \delta(n)$$
\end{proof}

\begin{theorem}[LOCC composition theorem][\Cref{thm:locccomp} restated]
Let $\mathcal{P},\mathcal{Q}$ be two idealized primitives, and let $C^{\mathcal{P}}$ be a (possibly stateful) construction $(\ell_\A,T_{\Sim},T_1,T_2,\delta)$-LOCC indifferentiable from $\mathcal{Q}$. Let $G$ be any LOCC cryptographic game making at most $T_{G,1}$ queries to its primitive and $T_{G,2}$ queries to its adversary, with at most $\ell_\A$ adversaries. Let $\Prim^{\mathcal{Q}}$ be any primitive relative to $\mathcal{Q}$ making at most $T_{\Prim}$ queries to its oracle. Let $\epsilon,T_{\A}:\N\to [0,1]$ be any functions and let $c$ be any constant.

As long as $T_{G,1}\cdot T_{\Prim} \leq T_2$ and $T_{G,2}\cdot T_{\A} \leq T_1$, then if $\Prim^{\cal{Q}}$ is $(T_{\A}\cdot T_{\Sim},\epsilon;c)$-secure under $G$ relative to $\cal{Q}$, then $\Prim^{C^{\cal{P}}}$ is $\left(T_{\A}, \epsilon + \delta;c\right)$-secure under $G$ relative to $\cal{P}$.
\end{theorem}

\begin{proof}
    This proof is exactly the same as~\Cref{thm:indiffcomp}. In particular, given any adversary \\$\A^{\mathcal{P}} = (\A_1^{\mathcal{P}},\dots,\A_{\ell_\A}^{\mathcal{P}})$, we define the adversary $\B^{\mathcal{Q}} = (\B_1^{\mathcal{Q}},\dots,\B_{\ell_\A}^{\mathcal{Q}})$ so that $\B_i^{\mathcal{Q}} = \A_i^{\Sim^{\mathcal{Q}}}$. In particular, each $\B_i^{\mathcal{Q}}$ initializes its own copy of the simulator $\Sim$.

    We then observe that the game $G(1^n,\Prim^{C^{\mathcal{P}}},\A^{\mathcal{P}})$ is an LOCC algorithm with query access to \\$((\mathcal{P},C_1^{\mathcal{P}}), \dots, (\mathcal{P},C_t^{\mathcal{P}}))$, while the game $G(1^n,\Prim^{\mathcal{Q}},\A^{\Sim^{\mathcal{Q}}})$ is an LOCC algorithm with query access to $((\Sim_1^{\mathcal{Q}},\mathcal{Q}), \dots, (\Sim_5^{\mathcal{Q}},\mathcal{Q}))$.

    The rest of the proof proceeds identically to~\Cref{thm:indiffcomp}, replacing indifferentiability with LOCC-indifferentiability.
\end{proof}
\section{Combinatorial Lemmas for~\Cref{sec:locc}}\label{sec:lemmas}

\begin{lemma}[\Cref{lem:zerosplit} restated]
    Let $A_1,A_2,B_1,B_2$ be registers containing $a_1,a_2,b_1,b_2$ states of dimension $N+1$ respectively. For all collision-free $T$,
    \begin{equation*}
        \begin{split}
            \ket{\Zero(T,(B_1,B_2))}_{A_1,A_2,B_1,B_2} \\
            = \sum_{\substack{X \subseteq T\\|X|\in [a_1,a_1+b_1]\\|T|-|X|\in [a_2,a_2+b_2]}} \sqrt{\alpha_{|X|}} \ket{\Zero(X,B_1)}_{A_1,B_1}\ket{\Zero(T\setminus X,B_2)}_{A_2,B_2}
        \end{split}
    \end{equation*}
    where
    $$\alpha_i = \frac{{b_1 \choose i-a_1}{b_2 \choose |T|-i-a_2}}{{b_1+b_2 \choose |T|-a_1-a_2}{|T| \choose i}}$$
\end{lemma}

\begin{proof}
    \begin{equation}
        \begin{split}
            \sum_{\substack{X \subseteq T\\|X|\in [a_1,a_1+b_1]\\|T|-|X|\in [a_2,a_2+b_2]}} \sqrt{\alpha_{|X|}} \ket{\Zero(X,B_1)}_{A_1,B_1}\ket{\Zero(T\setminus X,B_2)}_{A_2,B_2}\\
            = \sum_{\substack{X \subseteq T\\|X|\in [a_1,a_1+b_1]\\|T|-|X|\in [a_2,a_2+b_2]}}\sum_{\substack{|v|=a_2+b_1\\type(v) = X^{0(a_2+b_1)}\\v_i \neq 0\text{ for }i\in A_1}}\sum_{\substack{|w|=a_2+b_2\\type(w) = (T\setminus X)^{0(a_2+b_2)}\\w_i \neq 0\text{ for }i\in A_2}}\sqrt{\alpha_{|X|}}\frac{1}{\sqrt{|X|!{b_1 \choose |X|-a_2}}}\frac{1}{\sqrt{|T\setminus X|!{b_2 \choose |T\setminus X|-a_1}}} \ket{v}\ket{w}
        \end{split}
    \end{equation}

    But note that 
    \begin{equation}
        \begin{split}
            \sqrt{\alpha_{|X|}}\frac{1}{\sqrt{|T|!{b_1 \choose |T|-a_2}}}\frac{1}{\sqrt{|T|!{b_2 \choose |T|-a_2}}}\\
            =\sqrt{\frac{{b_1 \choose |X| - a_1}{b_2 \choose |T|-|X|-a_2}}{{b_1 + b_2 \choose |T| - a_1 - a_2}{|T| \choose |X|}}\frac{1}{{|X|!{b_1 \choose |X|-a_2}}}\frac{1}{{|T\setminus X|!{b_2 \choose |T\setminus X|-a_1}}}}\\
            =\sqrt{\frac{1}{{b_1 + b_2 \choose |T| - a_1 - a_2}\frac{|T|!}{|X|!(|T|-|X|)!}|X|!|T\setminus X|!}}\\
            =\sqrt{\frac{1}{{b_1 + b_2 \choose |T| - a_1 - a_2}|T|!}}\\
        \end{split}
    \end{equation}

    and so
    \begin{equation}
        \begin{split}
            \sum_{\substack{X \subseteq T\\|X|\in [a_1,a_1+b_1]\\|T|-|X|\in [a_2,a_2+b_2]}} \sqrt{\alpha_{|X|}} \ket{\Zero(X,B_1)}_{A_1,B_1}\ket{\Zero(T\setminus X,B_2)}_{A_2,B_2}\\
            = \sum_{\substack{X \subseteq T\\|X|\in [a_1,a_1+b_1]\\|T|-|X|\in [a_2,a_2+b_2]}}\sum_{\substack{|v|=a_2+b_1\\type(v) = X^{0(a_2+b_1)}\\v_i \neq 0\text{ for }i\in A_1}}\sum_{\substack{|w|=a_2+b_2\\type(w) = (T\setminus X)^{0(a_2+b_2)}\\w_i \neq 0\text{ for }i\in A_2}}\sqrt{\frac{1}{{b_1 + b_2 \choose |T| - a_1 - a_2}|T|!}}\ket{v}\ket{w}\\
            =\sum_{\substack{|v| = a_1+a_2+b_1+b_2\\type(v) = T^{0(a_1+a_2+b_1+b_2)}\\v_i\neq 0\text{ for }i\in A_1A_2}}\sqrt{\frac{1}{{b_1 + b_2 \choose |T| - a_1 - a_2}|T|!}}\ket{v}\\
            =\ket{\Zero(T, (B_1,B_2))}_{A_1A_2B_1B_2}
        \end{split}
    \end{equation}
\end{proof}

\begin{lemma}[\Cref{lem:zerosplit2} restated]
    Let $A_1,A_2,B_1,B_2$ be registers containing $a_1,a_2,b_1,b_2$ states of dimension $N+1$ respectively. For all collision-free $T$,
    \begin{equation}
        \begin{split}
            \ket{\Zero^2(T,B_1,B_2)^{b_1^f,b_2^f}}_{A_1,A_2,B_1,B_2}\\
            = \frac{1}{\sqrt{{|T|\choose a_1 + b_1^f}}}\sum_{\substack{X \subseteq T\\|X|=a_1+b_1^f}} \ket{\Zero(X,B_1)}_{A_1 B_1}\ket{\Zero(T\setminus X,B_2)}_{A_2 B_2}
        \end{split}
    \end{equation}
\end{lemma}

\begin{proof}
    \begin{equation}
        \begin{split}
            \frac{1}{\sqrt{{|T|\choose a_1 + b_1^f}}}\sum_{\substack{X \subseteq T\\|X|=a_1+b_1^f}} \ket{\Zero(X,B_1)}_{A_1 B_1}\ket{\Zero(T\setminus X,B_2)}_{A_2 B_2}\\
            =\mathfrak{N}\sum_{\substack{X \subseteq T\\|X|=a_1+b_1^f}} \sum_{\substack{|v|=a_1+b_1\\type(v)=X^{0(a_1+b_1)}\\v_i\neq 0\text{ for }i\in A_1}}\sum_{\substack{|w|=a_2+b_2\\type(v)=(T\setminus X)^{0(a_2+b_2)}\\w_i\neq 0\text{ for }i\in A_2}} \ket{v}\ket{w}\\
        \end{split}
    \end{equation}

    where
    \begin{equation}
        \begin{split}
            \mathfrak{N} = \sqrt{\frac{1}{{{|T|\choose a_1 + b_1^f}}{(a_1+b_1^f)!{b_1\choose b_1^f}}{(a_2+b_2^f)!{b_2\choose b_2^f}}}}\\
            =\sqrt{\frac{(a_1+b_1^f)!(|T|-a_1+b_1^f)!}{|T|!(a_1+b_1^f)!(a_2+b_2^f)!{b_1\choose b_1^f}!{b_2\choose b_2^f}!}}\\
            =\sqrt{\frac{(a_1+b_1^f)!(a_2+b_2^f)!}{|T|!(a_1+b_1^f)!(a_2+b_2^f)!{b_1\choose b_1^f}!{b_2\choose b_2^f}!}}\\
            =\sqrt{\frac{1}{|T|!{b_1\choose b_1^f}!{b_2\choose b_2^f}!}}
        \end{split}
    \end{equation}
    and so combining these two equations we get
    \begin{equation}
        \begin{split}
            \frac{1}{\sqrt{{|T|\choose a_1 + b_1^f}}}\sum_{\substack{X \subseteq T\\|X|=a_1+b_1^f}} \ket{\Zero(X,B_1)}_{A_1 B_1}\ket{\Zero(T\setminus X,B_2)}_{A_2 B_2}\\
            =\sqrt{\frac{1}{|T|!{b_1\choose b_1^f}!{b_2\choose b_2^f}!}}\sum_{\substack{X \subseteq T\\|X|=a_1+b_1^f}} \sum_{\substack{|v|=a_1+b_1\\type(v)=X^{0(a_1+b_1)}\\v_i\neq 0\text{ for }i\in A_1}}\sum_{\substack{|w|=a_2+b_2\\type(v)=(T\setminus X)^{0(a_2+b_2)}\\w_i\neq 0\text{ for }i\in A_2}} \ket{v}\ket{w}\\
            =\ket{\Zero^2(T,B_1,B_2)^{b_1^f,b_2^f}}_{A_1A_2B_1B_2}
        \end{split}
    \end{equation}
\end{proof}
\section{The spectrum of extended Kneser graphs}\label{app:kneser}
\begin{definition}
    For any $N,k\in \N$, the Kneser graph $K(N,k)$ is the graph with vertices subsets of $[N]$ of size $k$. There is an edge between $X,Y\subseteq[N]$ in $K(N,k)$ if and only if $X\cap Y = \emptyset$.
\end{definition}

\begin{definition}
    Define $A_{k,k'}^N$ to be the adjacency matrix with rows indexed by $X\subseteq [N]$ of size $k$ and columns indexed by $Y\subseteq [N]$ of size $k$ defined by
    $$A_{k,k'}^N[X,Y] = \begin{cases}
        1 & X\cap Y = \emptyset\\
        0 & X\cap Y\neq \emptyset
    \end{cases}$$

    It is clear that $A_{k,k}^N$ is the adjacency matrix of the Kneser graph $K(N,k)$.

    When clear from context, we will omit $N$.
\end{definition}

\begin{theorem}[\cite{kneserenergy}]\label{thm:kenergybasic}
    When $N\geq 2k + 1$, the trace norm of $A_{k,k}^N$ is 
    $$\sum_{j=0}^k \left({{N \choose j} - {N \choose j-1}}\right){N-k-j\choose k-j} = \frac{(N-1)(N-2)\dots(N-2k+1)2^k}{k!}$$
\end{theorem}

We will show the following extended version of this result
\begin{theorem}\label{thm:kneserfinal}
    For any $k_1 \leq k_2$, when $N\geq 2k_2 + 1$, the trace norm of $A_{k_1,k_2}^N$ is 
    $$\leq k_2 \sqrt{N^{k_1+k_2}}$$
\end{theorem}

To prove this, we will use a number of facts used in proving~\Cref{thm:kenergybasic}. In particular, we will follow the structure of the proof given in~\cite{kneserwayback}.

\begin{definition}
    Define $\wt{A}_{k,k'}^N$ to be the adjacency matrix with rows indexed by $X\subseteq [N]$ of size $k$ and columns indexed by $Y\subseteq [N]$ of size $k$ defined by
    $$\wt{A}_{k,k'}^N[X,Y] = \begin{cases}
        1 & X\subseteq Y\\
        0 & \text{otherwise}
    \end{cases}$$
    When clear from context, we will omit $N$.
\end{definition}

\begin{lemma}[Facts from~\cite{kneserwayback}]\label{lem:kfacts}
    Fix $N$
    \begin{equation}\label{eq:kf1}
        \wt{A}_{ij}\wt{A}_{jk} = {k - i \choose j - i}\wt{A}_{ik}
    \end{equation}
    \begin{equation}\label{eq:kf2}
        \wt{A}_{ij}{A}_{jk} = {N-k-i \choose j-i}A_{ik}
    \end{equation}
    \begin{equation}\label{eq:kf3}
        \wt{A}_{jk} = \sum_{i=0}^j (-1)^i {\wt{A}_{ij}}^\top A_{ik}
    \end{equation}
    \begin{equation}\label{eq:kf4}
        A_{jk} = \sum_{i=0}^j (-1)^i {\wt{A}_{ij}}^{\top} \wt{A}_{ik}
    \end{equation}
    \begin{equation}\label{eq:kf5}
        \text{For }i\leq j\leq k,\ \row(\wt{A}_{ik}) \subseteq \row(\wt{A}_{jk})
    \end{equation}
    \begin{equation}\label{eq:kf6}
        \text{For }i\leq j\leq N-k,\ \row({A_{ik}}) \subseteq \row({A}_{jk})
    \end{equation}
    \begin{equation}\label{eq:kf7}
        \text{For }j\leq k\leq N-j,\ \row({A_{jk}}) = \row(\wt{A}_{jk})
    \end{equation}
    \begin{equation}\label{eq:kf8}
        \text{For }t\leq k\leq N-t,\ \rank(A_{tk})=\rank(\wt{A}_{tk})={N\choose t}
    \end{equation}
\end{lemma}

We then proceed to the proof of~\Cref{thm:kneserfinal}

\begin{proof}
    Let $k_1\leq k_2$ and let $N \geq 2k_2 + 1$. We will explicitly compute the singular values of $A_{k_1,k_2}$. Recall that these are the eigenvalues of $A_{k_1k_2}^{\top} A_{k_1k_2} = A_{k_2k_1}A_{k_1k_2}$.

    In the style of~\cite{kneserwayback}, let $U_0$ be the row space of ${A}_{0k_2}$, and let $U_i$ be the orthogonal complement of $\row({A}_{(i-1)k_2})$ in $\row({A}_{ik_2})$. Note that by~\Cref{eq:kf5}, we know that $\row({A}_{(i-1)k_2}) \subseteq \row({A}_{ik_2})$.

    We then proceed to show that for each $0\leq j\leq k_1$, every $\vec{v}\in U_j$ is an eigenvector of $A_{k_2k_1}A_{k_1k_2}$ with eigenvalue
    $${N-k_1-j\choose k_2-j}{N-k_2-j\choose k_1-j}$$

    In particular, we know by~\Cref{eq:kf7} that $\vec{v}$ is in the row space of $\wt{A}_{jk_2}$, and so there exists a vector $\vec{a}$ such that $\vec{v}=\vec{a} \wt{A}_{jk_2}$. On the other hand, for all $i<j$, by~\Cref{eq:kf5,eq:kf7}, $\vec{v}$ is orthogonal to $A_{ik_2}$. And so we have
    \begin{equation}
        \vec{0}=\vec{u}{A}_{k_2i} = \vec{a} \wt{A}_{jk_2} A_{k_2i} = {N-k_2-i\choose j-i}\vec{a}A_{ji}
    \end{equation}
    and so $\vec{a}A_{ij}^{\top}=\vec{a}A_{ji}=\vec{0}$ for all $i<j$. But since $A_{ij}$ and $\wt{A}_{ij}$ have the same row space, we have $\vec{a}\wt{A}_{ij}^{\top} = \vec{0}$.

    We can then compute $\vec{u}A_{k_2k_1}A_{k_1k_2}$ as follows
    \begin{equation}
        \begin{split}
            \vec{u}A_{k_2k_1}A_{k_1k_2}\\
            =\vec{a}\wt{A}_{jk_2}A_{k_2k_1}A_{k_1k_2}\\
            =\vec{a}{N-k_1-j\choose k_2-j}A_{jk_1}A_{k_1k_2}\\
            ={N-k_1-j\choose k_2-j}\sum_{i=0}^j\left((-1)^i \vec{a}\wt{A}_{ij}^{\top}\wt{A}_{ik_1}\right)A_{k_1k_2}\\
            ={N-k_1-j\choose k_2-j}(-1)^j\vec{a} \wt{A}_{jj}^{\top} \wt{A}_{jk_1} A_{k_1k_2}\\
            ={N-k_1-j\choose k_2-j}(-1)^j\vec{a} \wt{A}_{jk_1} A_{k_1k_2}\\
            ={N-k_1-j\choose k_2-j}{N-k_2-j\choose k_1-j}(-1)^j\vec{a} A_{jk_2}\\
            =(-1)^j{N-k_1-j\choose k_2-j}{N-k_2-j\choose k_1-j} \vec{u}
        \end{split}
    \end{equation}

    By~\Cref{eq:kf8}, we know that the dimension of $U_i$ is ${N\choose j} - {N\choose j-1}$. We will show that the eigenvectors of $A_{k_2k_1}A_{k_1k_2}$ are precisely $U_0,\dots,U_{k_1}$. Thus, since $\rank(A_{k_1k_2}) = {N\choose k_1} = \sum_{j=0}^{k_1}\left({N\choose j} - {N\choose j-1}\right)$, these are all the non-zero eigenvectors.

    Thus, the singular values of $A_{k_1k_2}$ are, for each $0\leq j\leq k_1$,
    $$(-1)^j{N-k_1-j \choose k_2-j}{N-k_2-j \choose k_1-j}$$
    each with multiplicity ${N\choose j} - {N \choose j-1}$.

    Note that the trace norm of $A_{k_1k_2}$ is the sum of the square roots of the absolute values of the singular values. That is,
    \begin{equation}
        \begin{split}
            \norm{A_{k_1k_2}}_1\\
            =\sum_{j=0}^{k_1}\left({N\choose j} - {N \choose j-1}\right)\sqrt{{N-k_1-j \choose k_2-j}{N-k_2-j \choose k_1-j}}\\
            \leq \sum_{k=0}^{k_1} N^j \sqrt{N^{k_2-j+k_1-j}}\\
            = k_1 \sqrt{N^{k_1+k_2}}
        \end{split}
    \end{equation}
\end{proof}

\begin{corollary}[\Cref{cor:tauboundmain} restated]
    Let $\tau_{c_1,c_2,I,C}$ be the value defined in~\Cref{eq:taudef} in the proof of\label{lem:keylemma}. We have
    $$\norm{\tau_{c_1,c_2,I,C}}_1 \leq (a_1+b_1+b_2)N^{a_1+\frac{c_1+c_2}{2}-|I|}$$
\end{corollary}

\begin{proof}
    Observe that $\tau_{c_1,c_2,I,C}$ can be written as a block matrix, with all entries zero except for one which is equal to $A_{a_1+c_1-|I|,a_1+c_2-|I|}^{N-|I\cup C|}$. 

    Since $(a_1+a_2+b_1+b_2+1)^2\leq N$, and since $\max(c_1,c_2)\leq c\leq b_1+b_2$, it is clear that we have $2\max(a_1+c_1-|I|,a_1+c_2-|I|) + 1 \leq N$.

    Thus, by~\Cref{thm:kneserfinal},
    \begin{equation}
        \begin{split}
            \norm{\tau_{c_1,c_2,I,C}}_1 = \norm{A_{a_1+c_1-|I|,a_1+c_2-|I|}^{N-|I\cup C|}}_1\\
            \leq \max(a_1+c_1-|I|,a_1+c_2-|I|)N^{\frac{a_1+c_1-|I|+a_1+c_2-|I|}{2}}\\
            \leq (a_1+c)N^{a_1-|I|+\frac{c_1+c_2}{2}}
        \end{split}
    \end{equation}
\end{proof}
\section{More Details on Applications}
\subsection{OWSG do not exist in the $\Swap$ model}\label{sec:owsg}

\begin{definition}
    A one-way state generator is a cryptographic primitive $OWSG = (StateGen, Ver)$ with the following syntax
    \begin{enumerate}
        \item $StateGen(1^n, k) \to \rho_k$: on input the security parameter $n$ and a key $k \in \{0,1\}^n$, outputs a quantum state $\phi_k$.
        \item $Ver(1^n, k', \rho)$: on input the security parameter, a key $k'$, and a state $\rho$, outputs $0$ or $1$.
    \end{enumerate}
    and satisfying the following guarantees
    \begin{enumerate}
        \item Correctness: $(StateGen, Ver)$ is secure under the (adversary-less) game\\ $G_{OSWG}^{cor}(1^n, (StateGen, Ver))$, defined as follows:
        \begin{enumerate}
            \item Sample $k \gets \{0,1\}^n$.
            \item Run $StateGen(1^n, k) \to \rho_k$.
            \item Output $1$ if $Ver(k, {\rho_k})=0$.
        \end{enumerate}
        \item Security: $(StateGen, Ver)$ is secure under the game \\$G_{OWSG}^{sec-t}(1^n, (StateGen, Ver), \A), \A)$, defined as follows:
        \begin{enumerate}
            \item Sample $k \gets \{0,1\}^n$.
            \item Run $StateGen(1^n, k)$ $t+1$ times, producing ${\rho_k} \otimes {\phi_k}^{\otimes t}$.
            \item Run $\A(1^n, {\rho_k}^{\otimes t}) \to k'$.
            \item Output $Ver(k', {\rho_k})$.
        \end{enumerate}
    \end{enumerate}
\end{definition}

\begin{definition}
    An $(t,\epsilon)$-one-way state generator is a cryptographic primitive $OWSG = (StateGen, Ver)$ satisfying the same syntax as a OWSG but with the following weakened correctness and security requirements
    \begin{enumerate}
        \item Correctness: $(StateGen, Ver)$ is $\epsilon$-secure under the game $G_{OWSG}^{cor}$
        \item Security: $(StateGen, Ver)$ is $\epsilon$-secure under the game $G_{OWSG}^{sec-t}$
    \end{enumerate}
\end{definition}

\begin{lemma}[\Cref{lem:noowsg} restated]
    For any family of distributions $\D$, there does not exist a $\left(O(n^2), \frac{1}{10n^2}\right)$-one-way state generator relative to $\SSM$.
\end{lemma}

Since the $\SSMM$ model is just the $\SSM$ model under a different distribution of states, we get the following corollary.

\begin{corollary}
        For any family of distributions $\D$, there does not exist a $\left(10n, \frac{1}{10n^2}\right)$-one-way state generator relative to $\SSMM$.
\end{corollary}

\begin{lemma}
    If there exists a $\left(10n, \frac{1}{20n^2}\right)$-one-way state generator relative to $\Swap$, then there exists a $\left(10n, \frac{1}{10n^2}\right)$-one-way state generator relative to $\SSMM$.
\end{lemma}

This follows immediately from~\Cref{corf:swaptossmm} setting $d$ to $3$.

\begin{corollary}
    There does not exist a $\left(10n, \frac{1}{2n}\right)$-one-way state generator relative to $\Swap$.
\end{corollary}

\begin{corollary}
    There does not exist a one-way state generator relative to $\Swap$.
\end{corollary}
\subsection{Classical communication key exchange does not exist in the $\Swap$ model}\label{sec:ke}

\begin{definition}
    An $r_{max}$-classical communication key exchange protocol is a cryptographic primitive $(A,B)$ with the following syntax
    \begin{enumerate}
        \item $A(1^n,r,m_B) \to m_A$: a stateful procedure which takes in a security parameter $1^n$, a round number $r$, and Bob's message $m_B$, and produces a message $m_A$
        \item $B(1^n, r, m_A) \to m_B$: a stateful procedure which takes in a security parameter $1^n$, a round number $r$, and Alice's message $m_A$, and produces a message $m_B$
    \end{enumerate}
    satisfying the following properties
    \begin{enumerate}
        \item Correctness: $(A,B)$ is secure under the following (adversary-less) LOCC security game $G_{KE}^{cor}(1^n, (A,B)) = (G_1^A,G_2^B)$ defined as follows:
        \begin{enumerate}
            \item Set $m_B^0$ to be the empty string.
            \item For $i=1,\dots,r_{max}$
            \begin{enumerate}
                \item $G_1^A$ queries $A(1^n, i, m_B^{i-1}) \to m_A^i$ and forwards $m_A^i$ to $G_2^B$
                \item $G_2^B$ queries $B(1^n, i, m_A^i) \to m_B^i$ and forwards $m_B^i$ to $G_1^A$.
            \end{enumerate}
            \item $G_1^A$ receives $A$'s final output $b_A$.
            \item $G_2^B$ receives $B$'s final output $b_B$ and forwards the result to $G_1^A$.
            \item $G_1^A$ outputs $1$ if and only if $b_A\neq b_B$.
        \end{enumerate}
        \item Security: $(A,B)$ is $c=1/2$-secure under the following LOCC security game $G_{KE}^{sec}(1^n, (A,B)) = (G_1^A,G_2^B,G_3^\A)$.
        \begin{enumerate}
            \item Set $m_B^0$ to be the empty string.
            \item For $i=1,\dots,r_{max}$
            \begin{enumerate}
                \item $G_1^A$ queries $A(1^n, i, m_B^{i-1}) \to m_A^i$ and forwards $m_A^i$ to $G_2^B$ and $G_3^\A$.
                \item $G_2^B$ queries $B(1^n, i, m_A^i) \to m_B^i$ and forwards $m_B^i$ to $G_1^A$ and $G_3^\A$.
                \item $G_3^\A$ sends $m_A^i,m_B^i$ to $\A$.
            \end{enumerate}
            \item $G_3^\A$ receives $\A$'s final output $b'$, and forwards the result to $G_1^A$.
            \item $G_1^A$ receives $A$'s final output $b_A$.
            \item $G_1^A$ outputs $1$ if and only if $b_A = b'$.
        \end{enumerate}
    \end{enumerate}

    If $r_{max}\leq \poly(n)$, then we just say that $(A,B)$ is a classical communication key exchange protocol.
\end{definition}

\begin{definition}
    An $(r_{max},\epsilon)$-classical communication key exchange protocol is a protocol $(A,B)$ satisfying the same syntax as a classical communication key exchange protocol but with the following weakened correctness and security requirements
    \begin{enumerate}
        \item Correctness: $(StateGen, Ver)$ is secure under the game $G_{KE}^{cor}$
        \item Security: $(StateGen, Ver)$ is $(\epsilon,1/2)$-secure under the game $G_{KE}^{sec}$
    \end{enumerate}
\end{definition}

\begin{theorem}[Key exchange amplification]\label{thm:keamp}
    Let $r_{max} \leq \poly(n)$. If there exists a $(r_{max},1/n)$-classical communication key exchange protocol relative to any oracle $\mathcal{O}$, then there exists a $r_{max}$-classical communication key exchange protocol relative to $\mathcal{O}$.
\end{theorem}

For the remainder of this section, let $\D = \{\mu_n\}_{n\in \N}$ be the family of Haar random distributions on $n$ qubits. 

\begin{theorem}[Theorem 8.4 from~\cite{TCC:AnaGulLin24}]\label{thm:noke}
    There does not exist a classical communication key exchange protocol in the $\SSM$ model parameterized by $\D$.
\end{theorem}

\begin{corollary}
    There does not exist a classical communication key exchange protocol in the $\Swap$ model parameterized by $\D$.
\end{corollary}

\begin{proof}
    If there exists a protocol secure in the $\Swap$ model, then there also exists a $(r_{max},1/3n)$-classical communication key exchange protocol secure in the $\Swap$ model. And so by~\Cref{corf:swaptossmm} there exists a $(r_{max},1/2n)$ protocol secure in the $\SSMM$ model. But since the security games for classical communication key exchange are LOCC, by~\Cref{corf:ssmmtossm}, there exists a $(r_{max},1/n)$ protocol secure in the $\SSM$ model. By~\Cref{thm:keamp}, there exists a negligibly secure classical communication key exchange protocol secure in the $\SSM$ model.

    And so by~\Cref{thm:noke}, there does not exist a classical communication key exchange protocol in the $\Swap$ model.
\end{proof}

\end{document}